\documentclass[a4paper,11pt]{article}

\usepackage{hyperref}
 
\usepackage{a4wide}
 
\usepackage{tabularx}

\usepackage{amsmath}
\usepackage{amssymb}
\usepackage{amsthm} 

\usepackage{algorithm,algpseudocode}

\usepackage{filecontents}

\usepackage{tikz}
\usetikzlibrary{fit,positioning}
\newcommand\addvmargin[1]{
  \node[fit=(current bounding box),inner ysep=#1,inner xsep=0]{};
}
 
\theoremstyle{plain}
\newtheorem{thm}{Theorem} 
\newtheorem{cor}[thm]{Corollary}
\newtheorem{lem}[thm]{Lemma}
\newtheorem{obs}[thm]{Observation}
\newtheorem{prop}[thm]{Proposition}
\theoremstyle{definition}
\newtheorem{defn}[thm]{Definition}
 
 \newcommand{\Zivny}{{\v{Z}}ivn{\'y}}
 \def\Nesetril{Ne\v set\v ril}
 \def\Diaz{D\'iaz} 
\def\Lovasz{Lov{\'{a}}sz}

\newcommand{\Hom}[1]{\text{\#Hom$(#1)$}}
\newcommand{\Ret}[1]{\text{\#Ret$(#1)$}}
\newcommand{\SHom}[1]{\text{\#SHom$(#1)$}}
\newcommand{\Comp}[1]{\text{\#Comp$(#1)$}}
\newcommand{\LHom}[1]{\text{\#LHom$(#1)$}}
\newcommand{\LSHom}[1]{\text{\#LSHom$(#1)$}}
\newcommand{\LComp}[1]{\text{\#LComp$(#1)$}}

\newcommand{\DHom}[1]{\text{Hom$(#1)$}}
\newcommand{\DRet}[1]{\text{Ret$(#1)$}}
\newcommand{\DSHom}[1]{\text{SHom$(#1)$}}
\newcommand{\DComp}[1]{\text{Comp$(#1)$}}
\newcommand{\DLHom}[1]{\text{LHom$(#1)$}}

\newcommand{\SubSum}{\text{\#SubsetSum}}
\newcommand{\UniHom}{\text{Uniform\#HomToCliques}}
\newcommand{\UniSHom}{\text{Uniform\#SHomToCliques}}
 
\newcommand{\cHom}[1]{\text{\#Hom\textsuperscript{C}$(#1)$}}
\newcommand{\cComp}[1]{\text{\#Comp\textsuperscript{C}$(#1)$}}
\newcommand{\cLHom}[1]{\text{\#LHom\textsuperscript{C}$(#1)$}}

\newcommand{\cGSHom}[1]{\text{\#GraphSetHom\textsuperscript{C}$(#1)$}}
 
\newcommand{\prob}[3]{
  \begin{description}
    \item[\bf Name.] #1
    \vspace{-1ex}
    \item[\bf Input.] #2  
        \vspace{-1ex}
    \item[\bf Output.] #3
  \end{description}
}

\newcommand{\boldA}{\mathbf{A}}
\newcommand{\calA}{\mathcal{A}}
\newcommand{\boldb}{\mathbf{b}}
\newcommand{\calB}{\mathcal{B}}

\newcommand{\calC}{\mathcal{C}}

\newcommand{\calH}{\mathcal{H}}

\newcommand{\boldS}{\mathbf{S}}
\newcommand{\calS}{\mathcal{S}}

\newcommand{\boldv}{\mathbf{v}}
\newcommand{\boldx}{\mathbf{x}}
\newcommand{\Z}{\mathbb{Z}}

\newcommand{\FP}{\mathrm{FP}}
\newcommand{\numP}{\#\mathrm{P}}
\newcommand{\bis}{\#\text{BIS}}
\newcommand{\leap}{\le_\mathrm{AP}}

\newcommand{\abs}[1]{\left\vert #1 \right\vert}
\DeclareRobustCommand{\stirling}{\genfrac\{\}{0pt}{}}
\newcommand*\from{\colon}
\newlength{\templength}

\let\originalleft\left
\let\originalright\right
\renewcommand{\left}{\mathopen{}\mathclose\bgroup\originalleft}
\renewcommand{\right}{\aftergroup\egroup\originalright}

\renewcommand{\hom}[3][]{{N^{#1}\bigl(#2 \rightarrow #3\bigr)}}
\newcommand{\sur}[2]{\hom[\text{sur}]{#1}{#2}}
\newcommand{\comp}[2]{\hom[\text{comp}]{#1}{#2}}
\newcommand{\inj}[2]{\hom[\text{inj}]{#1}{#2}}
 
\newcommand{\ret}[2]{\hom[\text{ret}]{#1}{#2}}

\title{The Complexity of Counting Surjective Homomorphisms and Compactions}
\author{Jacob Focke, Leslie Ann Goldberg and  Stanislav \Zivny \thanks{
To appear in SIDMA. A short version of this paper (without the proofs) appeared in the proceedings of SODA 2018~\cite{FGZ2018}.
The research leading to these results has received funding from 
the European Research Council under the European Union's Seventh Framework Programme (FP7/2007-2013) ERC grant agreement no.\ 334828 and under the European Union's Horizon 2020 research and innovation programme (grant agreement No 714532). Jacob Focke has received funding from the Engineering and Physical Sciences Research Council (grant ref: EP/M508111/1). Stanislav \Zivny\ was supported by a Royal Society University Research Fellowship. The paper 
reflects only the authors' views and not the views of the ERC or the European Commission. The European Union is not liable for any use that may be made of the information contained therein.}}  
\date{9 April 2019} 
 
\begin{document}
\maketitle

\begin{abstract}

A homomorphism from a graph $G$ to a graph $H$ is a function from the vertices of $G$ to the vertices of $H$
that preserves edges. 
A homomorphism is \emph{surjective} if it uses all of the vertices of~$H$ and
it is a \emph{compaction} if it uses all of the vertices of~$H$ and all of the non-loop edges of~$H$.
Hell and \Nesetril\ gave  a complete characterisation of the complexity of
deciding whether there is a homomorphism from an input graph~$G$ to a fixed graph~$H$.
A complete characterisation is not known for surjective homomorphisms or for compactions, though there
are many interesting results.
Dyer and Greenhill gave a complete characterisation of the complexity of
counting homomorphisms from an input graph~$G$ to a fixed graph~$H$.
In this paper, we give a complete characterisation of the complexity of
counting surjective homomorphisms from an input graph~$G$ to a fixed graph~$H$
and we also give a complete characterisation of the complexity of
counting compactions from an input graph~$G$ to a fixed graph~$H$.
In an addendum we use our characterisations to point out a dichotomy for the complexity of the respective approximate counting problems (in the connected case). 
\end{abstract}

\section{Introduction} 

A homomorphism from a graph~$G$ to a graph~$H$
is a function from $V(G)$ to $V(H)$ that preserves edges. That is, the function maps every edge of~$G$ to an edge of~$H$.
Many structures in graphs, such as proper colourings, independent sets, and generalisations of these,
can be represented as homomorphisms, so the study of graph homomorphisms
has a long history in 
combinatorics~\cite{BorgsChayesLovaszCounting,BW,HellMiller,HNOld,HellNesetrilBook,LovaszBook}.

Much of the work on this problem is algorithmic in nature.  
A very important early work is Hell and \Nesetril's paper~\cite{HellNesetrilOriginal},
which gives a complete characterisation
of the complexity of the following decision problem, parameterised by a fixed graph~$H$:
``Given an input graph~$G$, determine whether there is a homomorphism from~$G$ to~$H$.''
Hell and \Nesetril\ showed that this problem can be solved in polynomial time if~$H$ has a loop or is loop-free and bipartite. They showed that it is NP-complete otherwise.
An important generalisation of the homomorphism decision problem is the
list-homomorphism decision problem.
Here, in addition to the graph~$G$, the input specifies, for each vertex $v$ of~$G$,
a list $S_v$ of permissible vertices of~$H$.
The problem is to determine whether there is a homomorphism from~$G$ to~$H$
that maps each vertex $v$ of~$G$ to a vertex in~$S_v$.
Feder, Hell and Huang~\cite{FHH} gave a complete characterisation of the
complexity of this problem.  This problem can be solved in polynomial time if~$H$ is a so-called bi-arc graph, 
and it is NP-complete otherwise.

More recent work has restricted attention to homomorphisms with certain properties.
A function from $V(G)$ to $V(H)$ is \emph{surjective}
if every element of $V(H)$ is  the image of at least one element of $V(G)$.
So a homomorphism from~$G$ to~$H$ is surjective if every vertex of~$H$ is ``used'' by the homomorphism.
There is still no complete characterisation 
of the complexity of determining whether there is a surjective homomorphism from an input graph~$G$ to
a graph~$H$, despite an impressive collection of results 
\cite{BodirskySurvey,GolovachNewHardness,GolovachFinding,GolovachTrees,MartinPaulusmaSHomC4}.
A homomorphism from $V(G)$ to $V(H)$ is a \emph{compaction} if
it uses every vertex of~$H$ and also every non-loop edge of~$H$ (so it is surjective both on $V(H)$ and on 
the non-loop edges in~$E(H)$).
Compactions have been studied under the name
``homomorphic image''~\cite{HellMiller,HellNesetrilBook}
and even under the name ``surjective homomorphism''~\cite{hombasis,LovaszBook}.
Once again, despite much work \cite{BodirskySurvey,VikasSICOMPCompl,Vikas2004,Vikas2004b,Vikas2005,Vikas2013}, there is still no characterisation
of the complexity of determining whether there is a compaction from an input graph~$G$ to
a graph~$H$.  
 
Dyer and Greenhill~\cite{DG} initiated the algorithmic study of \emph{counting} homomorphisms.
They gave a complete characterisation of the graph homomorphism counting problem, parameterised by a fixed graph~$H$:
``Given an input graph~$G$, determine how many homomorphisms there are from~$G$ to~$H$.''
Dyer and Greenhill showed that this problem can be solved in polynomial time if every component of~$H$
is a clique with all loops present or a biclique (complete bipartite graph) with no loops present. Otherwise,
the counting problem is $\numP$-complete. \Diaz, Serna and Thilikos \cite{Diaz2004}
and Hell and \Nesetril~\cite{Hell2004b} have shown that the same dichotomy characterisation holds for 
the problem of counting 
list homomorphisms.

The main contribution of this paper is to give complete dichotomy characterisations for the problems of
counting compactions and surjective homomorphisms.
Our main theorem, Theorem~\ref{thm:CompDichotomy}, shows that the characterisation for compactions is
different from the characterisation for counting homomorphisms.
If every component of $H$ is (i) a star with no loops present, (ii) a single vertex with a loop, or (iii) a single edge with two loops
then counting compactions to~$H$ is solvable in polynomial time. Otherwise, it is $\numP$-complete.
We also obtain the same dichotomy for the problem of counting list compactions.
Thus, even though the decision problem is still open for compactions, 
our theorem gives a complete classification of
the complexity of the corresponding 
counting problem.

There is evidence that computational problems involving  surjective homomorphisms
are more difficult than those involving (unrestricted) homomorphisms. 
For example, suppose that $H$ consists of a 3-vertex clique with no loops together with a single looped
vertex. As~\cite{BodirskySurvey} noted,
the problem of deciding whether there is a homomorphism from a loop-free input graph~$G$ to~$H$
is trivial (the answer is yes, since all vertices of~$G$ may be mapped to the loop) but 
the problem of determining whether there is a surjective homomorphism from a loop-free input graph~$G$ to~$H$
is NP-complete. (To see this, recall 
the NP-hard problem of determining whether a connected loop-free graph~$G'$ 
that is not bipartite is
$3$-colourable. Given such a graph $G'$, we may determine whether
it is 3-colourable by letting $G$ consist of the disjoint union of $G'$ and a loop-free  clique of size~$4$,
and then checking whether there is a surjective homomorphism from $G$ to~$H$.)
There is also evidence that \emph{counting} problems involving surjective homomorphisms
are more difficult than those involving unrestricted homomorphisms.
In Section~\ref{sec:UniSHom} 
we consider a \emph{uniform} homomorphism-counting problem 
where all connected components of $G$ are  cliques without loops and
all connected components of $H$ are  cliques with loops,
but both $G$ and $H$ are part of the input.
It turns out (Theorem~\ref{thm:Uniform}) that in this uniform case, counting homomorphisms is in $\FP$
but counting surjective homomorphisms is $\numP$-complete.
Despite this evidence,  we show (Theorem~\ref{thm:SHom}) that the
problem of counting surjective homomorphisms to a fixed graph~$H$ has the same complexity characterisation 
as the problem of counting all homomorphisms to~$H$: The problem is solvable in polynomial time if
every component of~$H$ is a  clique with loops or a  biclique without loops. Otherwise, 
it is $\numP$-complete. 
Once again, our dichotomy characterisation extends to the problem
of counting surjective list homomorphisms.
Even though the decision problem is still open for surjective homomorphisms, our theorem gives a complete
complexity classification of the corresponding counting problem.

In Section~\ref{sec:retract} we will introduce
one more related counting problem --- the problem of counting \emph{retractions}.
Informally, if $G$ is a graph containing an induced copy of~$H$ then a retraction from~$G$ to~$H$
is a homomorphism from~$G$ to~$H$ that maps the induced copy to itself.
Retractions are well-studied in combinatorics, often from an algorithmic perspective~\cite{BodirskySurvey,FHRetractions,FHH,FHJKN,Vikas2004,Vikas2005}.
A complexity classification is not known for the decision problem
(determining whether there is a retraction from an input to~$H$).
Nevertheless, it is easy to give a complexity characterisation for the corresponding counting problem (Corollary~\ref{cor:retdichotomy}).
This characterisation, together with our main results,
implies that a long-standing conjecture of Winkler about the complexity of the decision problems for compactions and retractions
is false in the counting setting. See Section~\ref{sec:retract} for details.

Finally, in an addendum to this work, we address the relaxed   versions 
of the counting problems where the goal is to \emph{approximately} count surjective homomorphisms, compactions and retractions. We use our theorems to give a complexity dichotomy in the connected case for all three of these problems.

\subsection{Notation and Theorem Statements}\label{sec:Thms}

In this paper graphs are undirected
and may contain loops.
A homomorphism from a graph~$G$ to 
a  graph~$H$ is a function $h \from V(G) \to V(H)$ 
such that, for all $\{u,v\} \in E(G)$, the image $\{h(u),h(v)\}$ is in $E(H)$.
We use $\hom{G}{H}$ to denote the number of homomorphisms from~$G$ to~$H$.
A homomorphism~$h$ is said to ``use'' a vertex $v\in V(H)$ if there
is a vertex $u\in V(G)$ such that $h(u)=v$.
It is \emph{surjective} if it uses every vertex of~$H$.
We use $\sur{G}{H}$ to denote the number of surjective homomorphisms from~$G$ to~$H$.
A homomorphism~$h$ is said to use an edge $\{v_1,v_2\} \in E(H)$ if
there is an edge $\{u_1,u_2\}\in E(G)$ such that $h(u_1)=v_1$ and $h(u_2)=v_2$.
It is a \emph{compaction} if it uses every vertex of~$H$ and every non-loop edge of~$H$.
We use $\comp{G}{H}$ to denote the number of compactions from~$G$ to~$H$.
$H$ is said to be  \emph{reflexive} if every vertex  has a loop.
It is said to be \emph{irreflexive} if no vertex   has a loop.
We study the following computational problems\footnote{The reason that the input graph $G$ is restricted to be irreflexive in these problems, 
but that $H$ is not restricted,
is that this is the convention in the literature.
Since our results will be complexity classifications, parameterised by~$H$, we strengthen the results by avoiding restrictions on~$H$.
Different conventions are possible regarding~$G$,  but hardness results are typically the most difficult part of the complexity classifications in this area, so restricting $G$ leads to technically-stronger results.}, which are parameterised by a graph~$H$.\bigskip

\noindent\begin{tabular}{lll} 
	\textbf{Name:} $\Hom{H}$. & \textbf{Name:} $\Comp{H}$. & \textbf{Name:}  $\SHom{H}.$\tabularnewline[-0.05cm]
	\textbf{Input:}  Irreflexive graph $G$. & \textbf{Input:} Irreflexive graph $G$. & \textbf{Input:}  Irreflexive graph $G$. \tabularnewline[-0.05cm]
	\textbf{Output:} $ \hom{G}{H}$. & \textbf{Output:}  $\comp{G}{H}$.  & \textbf{Output:}  $\sur{G}{H}$. \tabularnewline[-0.05cm]
\end{tabular}\bigskip
 
A \emph{list homomorphism} generalises a homomorphism in the same way that a list colouring of a graph generalises
a (proper) colouring.  
Suppose that $G$ is an irreflexive graph and that $H$ is a graph.
Consider a collection of sets
$\boldS=\{S_v\subseteq V(H)\,:\, v\in V(G)\}$
A \emph{list homomorphism} from
$(G,\boldS)$ to $H$
is a homomorphism~$h$ from~$G$ to~$H$
such that, for every vertex $v$ of $G$, $h(v)\in S_v$. 
The set $S_v$ is referred to as a ``list'', specifying the allowable targets of vertex~$v$.
We use $\hom{(G,\boldS)}{H}$ to denote the number of list homomorphisms from~$(G,\boldS)$ to~$H$,
$\sur{(G,\boldS)}{H}$ to denote the number of surjective list homomorphisms from~$(G,\boldS)$ to~$H$
and $\comp{(G,\boldS)}{H}$ to denote the number of list homomorphisms from~$(G,\boldS)$ to~$H$  that
are compactions. We study the following additional computational problems, again parameterised by a graph~$H$.
\bigskip

\noindent\begin{tabular}{lll} 
\textbf{Name:} $\LHom{H}$. & \textbf{Name:} $\LComp{H}$. & \textbf{Name:}  $\LSHom{H}$.
\tabularnewline[-0.05cm]
\textbf{Input:}  Irreflexive graph $G$ & \textbf{Input:} Irreflexive graph $G$ & \textbf{Input:}  Irreflexive graph $G$
\tabularnewline[-0.05cm]
and a collection of lists & and a collection of lists & and a collection of lists
\tabularnewline[-0.05cm]
	$\boldS=\{S_v\subseteq V(H)\,:\, v\in V(G)\}$. &
		$\boldS=\{S_v\subseteq V(H)\,:\, v\in V(G)\}$. &
			$\boldS=\{S_v\subseteq V(H)\,:\, v\in V(G)\}$.
\tabularnewline[-0.05cm]
\textbf{Output:} $ \hom{(G,\boldS)}{H}$. & \textbf{Output:}  $\comp{(G,\boldS)}{H}$.  & \textbf{Output:}  $\sur{(G,\boldS)}{H}$. \tabularnewline[-0.05cm]
\end{tabular}\bigskip

In order to state our theorems, we define some classes of graphs.
A graph~$H$ is a \emph{clique} if, for every pair $(u,v)$ of distinct vertices, 
$E(H)$ contains the edge $\{u,v\}$. 
(Like other graphs, cliques may contain loops but not all loops need to be present.)
$H$ is a \emph{biclique} if it is bipartite (disregarding any loops) and there is a partition of $V(H)$ into two disjoint sets~$U$ and~$V$
such that, for every $u\in U$ and $v\in V$, $E(H)$ contains the edge $\{u,v\}$.
A biclique is a \emph{star} if $|U|=1$ or $|V|=1$ (or both).
Note that a star may  have only one vertex since, for example, we could have $|U|=1$ and $|V|=0$.
We sometimes use the notation $K_{a,b}$ to denote an irreflexive biclique whose vertices can be partitioned into~$U$ and~$V$
with $|U|=a$ and $|V|=b$.
The size of a graph is the number of vertices that it has.
We can now state the theorem of Dyer and Greenhill~\cite{DG}, as extended to list homomorphisms
by \Diaz, Serna and Thilikos \cite{Diaz2004}
and Hell and \Nesetril~\cite{Hell2004b}.

\begin{thm}[Dyer, Greenhill]\label{thm:DGorig}
Let $H$ be a graph. If every connected component of $H$ is  a reflexive  clique or an irreflexive  biclique, 
then $\Hom{H}$ and $\LHom{H}$ are in $\FP$. 
Otherwise, $\Hom{H}$ and $\LHom{H}$ are $\numP$-complete.
\end{thm} 

We can also state the main results of this paper.

\newcommand{\ThmComp}{
Let $H$ be a graph. If every connected component of $H$ is an irreflexive star
or a reflexive clique of size at most~$2$  then 
$\Comp{H}$ and $\LComp{H}$ are in $\FP$. 
Otherwise, $\Comp{H}$ and $\LComp{H}$ are  $\numP$-complete.
}\begin{thm}\label{thm:CompDichotomy}
\ThmComp
\end{thm}

\newcommand{\ThmSHom}{
Let $H$ be a graph. If every connected component of $H$ is  a reflexive  clique or an irreflexive  biclique,  
then $\SHom{H}$ and $\LSHom{H}$ are in $\FP$. 
Otherwise, $\SHom{H}$ and $\LSHom{H}$ are $\numP$-complete.
}
\begin{thm}\label{thm:SHom}
\ThmSHom
\end{thm}

The tractability results in Theorem~\ref{thm:CompDichotomy} follow from the fact that the number of compactions from a  graph $G$ to a graph $H$ can be expressed as a linear combination of the number of homomorphisms from $G$ to certain subgraphs of $H$, see Section~\ref{sec:easyComp}.
A proof sketch of the intractability result in Theorem~\ref{thm:CompDichotomy} is given at the beginning of Section~\ref{sec:hardComp}.
Theorem~\ref{thm:SHom} is simpler, see Section~\ref{sec:SHom}.

\subsection{Reductions and Retractions}\label{sec:retract}

In the context of two computational problems $\mathrm{P}_1$ and $\mathrm{P}_2$, we write $\mathrm{P}_1\le \mathrm{P}_2$ if there exists a polynomial-time Turing reduction from $\mathrm{P}_1$ to $\mathrm{P}_2$. If there exist such reductions in both directions, we write $\mathrm{P}_1\equiv\mathrm{P}_2$. 
Theorems~\ref{thm:DGorig}, \ref{thm:CompDichotomy} and \ref{thm:SHom}
imply the following observation.

\begin{obs}
\label{obs:DG}
Let $H$ be a graph. Then
$$\Hom{H} \equiv \LHom{H} \equiv \SHom{H} \equiv \LSHom{H} \leq \Comp{H} \equiv \LComp{H}.$$
\end{obs}

In order to see how Observation~\ref{obs:DG} contrasts with 
the situation concerning decision problems, it is useful to define decision 
versions of the computational problems that we study. Thus, $\DHom{H}$ is the problem of determining whether
$\hom{G}{H}=0$, given an input~$G$ of~$\Hom{H}$. The decision problems $\DComp{H}$, $\DSHom{H}$ and
$\DLHom{H}$   are defined similarly.

It is also useful to define the notion of a \emph{retraction}.
Suppose that $H$ is a graph 
with $V(H)= \{v_1,\ldots,v_c\}$
and that $G$ is an irreflexive graph.
We say that a tuple $(u_1,\ldots,u_c)$ of $c$ distinct vertices of $G$
induces a copy of~$H$ if, for every 
$1 \leq a < b \leq c$, 
$\{u_a,u_b\} \in E(G) \Longleftrightarrow \{v_a,v_b\} \in E(H)$.
A \emph{retraction}  from $(G;u_1,\ldots,u_c)$ to $H$
is a homomorphism~$h$ from~$G$ to~$H$
such that, for all $i\in[c]$, $h(u_i)=v_i$.
We use $\ret{(G;u_1,\ldots,u_c)}{H}$ to denote the number of retractions from
$(G;u_1,\ldots,u_c)$ to~$H$. We briefly consider 
the retraction counting and decision problems, which are parameterised by
a graph~$H$ with $V(H)=\{v_1,\ldots,v_c\}$.\footnote{Once again, some works would allow $G$ to have loops, and
would insist that loops are preserved in the induced copy of~$H$. We prefer to stick with the convention that $G$ is
irreflexive, but this does not make a difference to the complexity classifications that we describe.}
\bigskip

\noindent\begin{tabular}{ll} 
\textbf{Name:} $\Ret{H}$. & \textbf{Name:} $\DRet{H}$.  
\tabularnewline[-0.05cm]
\textbf{Input:}  Irreflexive graph $G$ and a tuple& \textbf{Input:} Irreflexive graph $G$  and a tuple
\tabularnewline[-0.05cm]
$(u_1,\ldots,u_c)$ of distinct vertices of $G$ & $(u_1,\ldots,u_c)$ of distinct vertices of $G$ 
\tabularnewline[-0.05cm]
that induces a copy of $H$. & that induces a copy of $H$.
\tabularnewline[-0.05cm]
\textbf{Output:} $ \ret{(G;u_1,\ldots,u_c)}{H}$. & \textbf{Output:}  Does $\ret{(G;u_1,\ldots,u_c)}{H}=0$?    \tabularnewline[-0.05cm]
\end{tabular}\bigskip

The following observation appears as Proposition~1 of \cite{BodirskySurvey}.
The proposition is stated for more general structures than graphs, but it applies equally to our setting. 
\begin{prop}[Bodirsky et al.]
Let $H$ be a graph. Then
$$\DHom{H} \leq \DSHom{H} \leq \DComp{H} \leq \DRet{H} \leq \DLHom{H}.$$
\end{prop}

We have already mentioned the fact (pointed out by Bodirsky et al.) that if $H$ is 
an irreflexive
$3$-vertex clique  together with a single looped vertex,
then 
$\DHom{H}$ is in P, but  $\DSHom{H}$ is NP-complete.  
There are no known graphs~$H$ separating $\DSHom{H}$, $\DComp{H}$ and $\DRet{H}$.
Moreover, Bodirsky et al.\ mention a conjecture~\cite[Conjecture 2]{BodirskySurvey},  
attributed to Peter Winkler, that, for all graphs~$H$, $\DComp{H}$ and $\DRet{H}$ 
are polynomially Turing equivalent.

The following observation, together with our theorems,
implies Corollary~\ref{cor:reductions} (below), 
which shows that the generalisation of Winkler's conjecture to the counting
setting is false unless $\FP=\numP$, since   $\Comp{H}$ and $\Ret{H}$
are not polynomially Turing equivalent for all $H$.

\begin{obs}
\label{obs:ret}
Let $H$ be a graph. Then
$\Ret{H}\leq \LHom{H}$ and $\Hom{H} \leq \Ret{H}$
\end{obs}
\begin{proof}
Let $V(H)=\{v_1,\ldots,v_c\}$.
We first reduce $\Ret{H}$ to $\LHom{H}$.
Consider an input to $\Ret{H}$ consisting of 
$G$ and $(u_1,\ldots,u_c)$.
For each $a\in[c]$, let $S_{u_a}$ be the set containing 
the single vertex $v_a$.
For each $v \in V(G)\setminus \{u_1,\ldots,u_c\}$, let $S_v = V(H)$.
Let $\boldS=\{S_v \,:\, v\in V(G)\}$.
Then $\ret{(G;u_1,\ldots,u_c)}{H} = \hom{(G,\boldS)}{H}$.

We next reduce $\Hom{H}$ to $\Ret{H}$.
Let $E^0$ be the set of all non-loop edges of~$H$.
Consider an input~$G$ to~$\Hom{H}$. Suppose without loss of generality that $V(G)$ is disjoint from
$V(H)=\{v_1,\ldots,v_c\}$. 
Let $G'$ be the graph with vertex set $V(G) \cup V(H)$
and edge set $E(G) \cup E^0$. Then $(v_1,\ldots,v_c)$ induces a copy of~$H$ in~$G'$ and
$\hom{G}{H} = \ret{(G';v_1,\ldots,v_c)}{H}$.
\end{proof}

Observation~\ref{obs:ret} immediately implies the following dichotomy characterisation for the problem
of counting retractions.

\begin{cor}\label{cor:retdichotomy}
Let $H$ be a graph. If every connected component of $H$ is  a reflexive  clique or an irreflexive  biclique,  
then $\Ret{H}$  is in $\FP$. 
Otherwise, $\Ret{H}$  is $\numP$-complete.
\end{cor}
\begin{proof}
The corollary follows immediately from Observation~\ref{obs:ret} and Theorem~\ref{thm:DGorig}.
\end{proof}

\begin{cor}\label{cor:reductions}
Let $H$ be a graph. Then
$$\Hom{H} \equiv \LHom{H} \equiv \SHom{H} \equiv \LSHom{H} 
\equiv \Ret{H}
\leq \Comp{H} \equiv \LComp{H}.$$
Furthermore, there is a graph~$H$ for which $\Comp{H}$ and $\LComp{H}$ are $\numP$-complete, but
$\Hom{H}$, $\LHom{H}$, $\SHom{H}$, $\LSHom{H}$ and $\Ret{H}$ are in $\FP$.
\end{cor}
\begin{proof} Theorems~\ref{thm:DGorig}, \ref{thm:CompDichotomy}, \ref{thm:SHom} and 
Corollary~\ref{cor:retdichotomy} give complexity classifications for all of the problems.
The reductions in the corollary follow from three easy observations.
\begin{itemize}
\item  All problems in $\FP$ are trivially inter-reducible.
\item All  $\numP$-complete problems are inter-reducible. 
\item  All problems in $\FP$ are reducible to all $\numP$-complete problems.
\end{itemize}
The separating graph~$H$ can be taken to be any reflexive clique of size at least~$3$ or any irreflexive biclique
that is not a star.
\end{proof}

\subsection{Related Work}

This section was added after the announcement of our results
(\url{https://arxiv.org/abs/1706.08786v1}), in order to 
draw attention  to some interesting subsequent work~\cite{HolgerDell, HubieChen}.

Both our tractability results and our hardness results  
rely on the fact  (see Theorem~\ref{thm:CompDecomp})
that the number of compactions from~$G$ to~$H$ can be expressed as a linear combination
of the number of homomorphisms from~$G$ to  certain subgraphs $J$ of $H$.
A similar statement applies to surjective homomorphisms.

As we note in the paper, these kinds of linear combinations have been  noticed 
in related contexts before, for example
in \cite[Lemma 4.2]{Borgs} and in \cite{LovaszBook}.
We use  the linear combination of Theorem~\ref{thm:CompDecomp},  together with interpolation, to prove hardness.
Although it is standard to restrict the input graph~$G$ to be irreflexive
(and this restriction makes the results stronger) the fact that $G$ is required
to be irreflexive causes severe difficulties.
 
In fact, Dell's note about our paper \cite{HolgerDell} 
shows that,  if you weaken the theorem statements by allowing the input $G$ to
have loops, then a simpler interpolation
based on a very
 recent paper  by Curticapean, Dell and Marx
\cite{hombasis} 
can be used to make the proofs very elegant! 
The exact same idea, written  more generally, was also discovered by Chen~\cite{HubieChen}.

\section{Preliminaries} \label{sec:Prelim}

It will often be technically convenient to restrict the problems that we study by requiring the
input graph~$G$ to be connected. In each case, we do this by adding a superscript ``$C$''
to the name of the problem.
For example, the problem $\cHom{H}$ is defined as follows.
\prob{\cHom{H}.}{A \emph{connected} irreflexive graph $G$.}{$\hom{G}{H}$.}
It is well known and easy to see
(See, e.g., \cite[(5.28)]{LovaszBook}) that if $G$ is an irreflexive graph with 
components $G_1,\ldots,G_t$ 
then $\hom{G}{H} = \prod_{i\in[t]} \hom{G_i}{H}$. Similarly, given $\boldS=\{S_v\subseteq V(H) : v\in V(G)\}$ let $\boldS_i=\{S_v : v\in V(G_i)\}$. Then $\hom{(G,\boldS)}{H}= \prod_{i\in[t]} \hom{(G_i,\boldS_i)}{H}$.
Thus, Dyer and Greenhill's theorem (Theorem~\ref{thm:DGorig})
can be re-stated in the following convenient form.

\begin{thm}[Dyer, Greenhill]\label{thm:DG}
Let $H$ be a graph. 
If every connected component of $H$ is   a reflexive  clique or an irreflexive  biclique, 
then $\cHom{H}$, $\Hom{H}$, $\cLHom{H}$ and $\LHom{H}$  are all in  $\FP$. 
Otherwise, 
$\cHom{H}$, $\Hom{H}$, $\cLHom{H}$ and $\LHom{H}$ are  all $\numP$-complete.
\end{thm} 
 
Finally, we introduce some frequently used notation. 
For every positive integer~$n$, we define $[n]= \{1,\dots,n\}$. 

A subgraph $H'$ of $H$ is said to be \emph{loop-hereditary} with respect to~$H$
if for every $v\in V(H')$ that is contained in a loop in~$E(H)$, $v$ is also contained in a loop in~$E(H')$.
 
We indicate that two graphs~$G_1$ and~$G_2$ are isomorphic by 
writing  $G_1\cong G_2$.

Given sets $S_1$ and $S_2$, we write $S_1 \oplus S_2$ for the disjoint union
of $S_1$ and $S_2$.
Given graphs $G_1$ and $G_2$, we write $G_1 \oplus G_2$ for the graph
$(V(G_1)\oplus V(G_2), E(G_1) \oplus E(G_2))$.
If $V$ is a set of vertices 
then we write 
$G_1 \oplus V$ as shorthand for the graph
$G_1 \oplus (V,\emptyset)$.
Similarly, if $M$ is a matching (a set of disjoint edges)
with vertex set $V$,
then we write
$G_1 \oplus M$ as shorthand for the graph
$G_1 \oplus (V,M)$.

\section{Counting Compactions}\label{sec:Comp}
  
The section is divided into a short subsection on tractable cases and the main subsection on hardness results which also contains 
the proof of the final dichotomy result, Theorem~\ref{thm:CompDichotomy}.

\subsection{Tractability Results}\label{sec:easyComp}

The tractability result in Lemma~\ref{lem:LCompFP} 
follows from the fact (see Theorem~\ref{thm:CompDecomp})
that the number of compactions from~$G$ to~$H$ can be expressed as a linear combination
of the number of homomorphisms from~$G$ to  certain subgraphs $J$ of $H$.
While we need the full details of   our particular linear expansion to  derive our hardness results,
the following simpler version 
suffices for tractability.

\begin{lem}\label{lem:LCompFP}
Let $H$ be a graph such that every connected component   is an irreflexive star
or a reflexive clique of size at most~$2$. Then \Comp{H} and \LComp{H} are in $\FP$. 
\end{lem}
\begin{proof}  First we deal with the case that $H$ is the empty graph.
Suppose that $H$ is  the empty graph and let
$(G,\boldS)$ be an instance of \LComp{H}. 
If  $G$ is empty then $\comp{(G,\boldS)}{H}=1$.
Otherwise,  $\comp{(G,\boldS)}{H}=0$.
Thus, if $H$ is empty, then \LComp{H} is in $\FP$. 
Obviously, this also implies that \Comp{H} is in $\FP$.

Let $\calH$ be the set of all non-empty graphs in which every connected component is
an irreflexive star or a reflexive clique of size at most~$2$.
We will show  that for every $H\in \calH$, 
\LComp{H} is in $\FP$. To do this, we need
the following notation.
Given a graph $H$, let
$m(H)$ denote the sum of $|V(H)|$ and the number of non-loop edges of~$H$.
We will use induction on $m(H)$.

The base case is $m(H)=1$. In this case, $H$ has only one vertex~$w$.
If $G$ is non-empty and has $w\in S_v$ for every vertex $v\in V(G)$
then $\comp{(G,\boldS)}{H}=1$.
Otherwise, 
$\comp{(G,\boldS)}{H}=0$. So \LComp{H} is in $\FP$.

For the inductive step, consider some $H\in \calH$ with $m(H)>1$.  
Let $(G,\boldS)$ be an instance of \LComp{H}.
If $G$ is empty then $\comp{(G,\boldS)}{H}=0$, so suppose that $G$ is non-empty.
For every subgraph $H'$ of $H$ let
$\boldS_{H'}$ denote the set of lists
$\boldS_{H'} = \{S_v \cap V(H') : v \in V(G)\}$.  
It  is easy to see that 
$\hom{(G,\boldS)}{H} =  
\sum_{H'  } \comp{(G,\boldS_{H'})}{H'},$ where the sum is over all loop-hereditary subgraphs $H'$ of $H$.
This observation is well known and is implicit, e.g, in the proof of a lemma of Borgs, Chayes, Kahn and \Lovasz~\cite[Lemma 4.2]{Borgs} (in a context without lists or loops).

A subgraph $H'$ of~$H$ is said to be a \emph{proper} subgraph of~$H$ if either 
$V(H')$ is a strict subset of $V(H)$ or $E(H')$ is a strict subset of $E(H)$ (or both). 
For every graph~$H$, let $Sub^<(H)$ denote the set of 
non-empty proper subgraphs of~$H$ that are loop-hereditary with respect to~$H$. 
Note that if $H\in \calH$ and $H' \in Sub^<(H)$ then $H'\in \calH$ and $m(H')<m(H)$.
We can refine the summation as follows.

\[\hom{(G,\boldS)}{H} = \comp{(G, \boldS)}{H} + 
\sum_{H' \in Sub^<(H)} \comp{(G,\boldS_{H'})}{H'}.\]
Since $H\in \calH$, every component of~$H$ is 
a reflexive clique or an irreflexive biclique,
so Theorem~\ref{thm:DGorig} shows that the 
quantity
$\hom{(G,\boldS)}{H}$ on the left-hand side
can be computed in polynomial time.
By induction, every term of the form $\comp{(G,\boldS_{H'})}{H'}$
can also be computed in polynomial time.
Subtracting this from the left-hand side, we obtain
$\comp{(G, \boldS)}{H}$, as desired. 

Thus, we have proved that $\LComp{H}$ is in $\FP$.
The problem $\Comp{H}$ is a restriction of $\LComp{H}$, so it is also in $\FP$.

\end{proof}

\subsection{Hardness Results}\label{sec:hardComp} 

This is the key section of this work. 
In this section, we consider a graph $H$ that has a connected component
that is not an irreflexive star or a reflexive clique of size at most~$2$.
The objective is to show that \Comp{H} and \LComp{H} are $\numP$-hard (this is the hardness 
content of Theorem~\ref{thm:CompDichotomy}).

We start with a brief proof sketch.
The  easy case is when $H$ contains a component that is not a reflexive clique or an irreflexive biclique.
In this case, Dyer and Greenhill's Theorem~\ref{thm:DGorig} shows that $\Hom{H}$ is $\numP$-hard.
We obtain the desired hardness by giving (in Theorem~\ref{thm:HomToComp}) a polynomial-time
Turing reduction from $\Hom{H}$ to $\Comp{H}$.
The result is finished off with a trivial reduction from $\Comp{H}$ to $\LComp{H}$.
The proof of Theorem~\ref{thm:HomToComp} is not difficult ---  
given an input $G$ to $\Hom{H}$, we add isolated vertices and edges to $G$
and recover the desired quantity $\hom{G}{H}$ using an oracle for \Comp{H} 
and polynomial interpolation.
There are small technical issues related to size-$1$ components in~$H$, and these are dealt with in Lemma~\ref{lem:HomToComp1}.

The more interesting case is when every component of $H$ is a reflexive clique or an irreflexive biclique, but
some component is either a reflexive clique of size at least~$3$ or an irreflexive biclique that is not a star.
The first milestone is Lemma~\ref{thm:simpleCompDichotomy}, which shows $\numP$-hardness
in the special case where $H$ is connected.
We prove Lemma~\ref{thm:simpleCompDichotomy} in a slightly stronger setting where the input graph
$G$ is connected. This allows us, in the remainder of the section, to generalise the connected case to the case in which $H$
is not connected.

The main difficulty, then, is Lemma~\ref{thm:simpleCompDichotomy}. The goal is to show that $\Comp{H}$ is $\numP$-hard when 
$H$ is a reflexive clique of size at least~$3$ or an irreflexive biclique that is not a star. 
Our main method for solving this
problem is   a technique (Theorem~\ref{thm:CompDecomp})
that lets us compute the number of compactions from a connected graph $G$ to a connected graph $H$ 
using a weighted sum of homomorphism counts,
say $\hom{G}{J_1},\ldots,\hom{G}{J_k}$. 
An important feature is that some of the weights might be negative.

Our basic approach will be to find a constituent $J_i$ such that $\cHom{J_i}$ is $\numP$-hard
and to reduce $\cHom{J_i}$ to the problem of computing the weighted sum.
Of course, if computing $\hom{G}{J_1}$ is $\numP$-hard and 
computing $\hom{G}{J_2}$ is $\numP$-hard, it does not follow that computing a weighted sum of these is $\numP$-hard.

In order to solve this problem, in Lemmas~\ref{lem:Lovasz1} 
and~\ref{lem:Lovasz2} we use an argument similar to that of Lov\'asz~\cite[Theorem 3.6]{Lovasz1967} 
to prove the existence of input instances that help us to distinguish between the problems 
$\cHom{J_1}, \ldots, \cHom{J_k}$.
Theorem~\ref{thm:CompInterpol1} 
then provides the desired reduction from 
a  chosen $\cHom{J_i}$ to the problem of computing the weighted sum.
Theorem~\ref{thm:CompInterpol1} is proved by a more complicated interpolation construction, in which 
we use the instances from Lemma~\ref{lem:Lovasz2} to modify the input.

Having sketched the proof at a high level, we are now ready to begin.
We start by working
 towards the proof of Theorem~\ref{thm:HomToComp}. The first step is to show that deleting 
size-$1$ components from $H$ does not add any complexity to \Comp{H}.

\begin{lem}\label{lem:HomToComp1}
Let $H$ be a graph that has exactly $q$  size-$1$ components.
Let $H'$ be the graph  constructed from $H$ by removing all  size-$1$ components. Then
$\Comp{H'} \le \Comp{H}$.
\end{lem}
\begin{proof}
Let $W=\{w_1, \dots, w_q\}$ 
be the vertices of $H$ that are contained in size-$1$ components. We can assume $q\ge 1$, otherwise $H'=H$. Let $G'$ be an input to \Comp{H'} and note that $G'$ might contain isolated vertices. For any non-negative integer~$t$, let $V_t$ be a set of $t$ isolated vertices,
distinct from the vertices of~$G'$,
and let $G_t = G' \oplus V_t$. For all $i\in \{0,\dots, t\}$, we define $S^i(G')$ to be the number of homomorphisms $\sigma$ from $G'$ to $H$ with the following properties:
\begin{enumerate}
\item $\sigma$ uses all non-loop edges of $H'$.
\item $\abs{\sigma(V(G'))\cap \{w_1, \dots, w_q\}} = i$,
\end{enumerate}
where $\sigma(V(G'))$ is the  image of $V(G')$ under the map~$\sigma$. 
We define $N^i(V_t)$ as the number of homomorphisms $\tau$ from $V_t$ to $H$ such that $\{w_1, \dots, w_i\} \subseteq \tau(V(V_t))$. 
Intuitively, $N^i(V_t)$ is the number of homomorphisms from $V_t$ to $H$ that use at least a set of $i$ arbitrary but fixed vertices of $H$, as the particular choice of vertices $\{w_1, \dots, w_i\}$ is not important when counting homomorphisms from a set of isolated vertices.
For any compaction $\gamma \from V(G_t) \to V(H)$, the restriction
$\gamma\vert_{V(G)}$ has to use all non-loop edges in $H'$. As $H'$ does not have size-$1$ components, this implies that all vertices other than $w_1, \dots, w_q$ are used by $\gamma\vert_{V(G)}$. Say, additionally, that $\gamma$ uses $q-i$ vertices from $W$, for some $i\in \{0,\dots,q\}$. Then, $\gamma\vert_{V_t}$ has to use the remaining $i$ vertices.
Thus, for each fixed $t\ge 0$, we obtain a linear equation:
\[
\underbrace{\comp{G_t}{H}}_{b_t} = \sum_{i=0}^q \underbrace{S^{q-i}(G')}_{x_i}\underbrace{N^i(V_t)}_{a_{t,i}}.
\]

By choosing $q+1$ different values for the parameter $t$ we obtain a system of linear equations. Here, we choose $t=0,\dots ,q$. Then the system is of the form $\boldb= \boldA\boldx$ for 
\begin{equation*}
\boldb=
\begin{pmatrix}
b_0\\
\vdots\\
b_q
\end{pmatrix}
\qquad
\boldA=
\begin{pmatrix}
a_{0,0} &\dots & a_{0,q}\\
\vdots & \ddots & \vdots\\
a_{q,0} &\dots & a_{q,q}\\
\end{pmatrix}
\qquad
\text{and}
\qquad
\boldx= \begin{pmatrix}
x_0\\
\vdots\\
x_q
\end{pmatrix}.
\end{equation*}
Note, that the vector $\boldb$ can be computed using $q+1$ \Comp{H} oracle calls. Further,
\[x_{q} = S^{0}(G') = \comp{G'}{H'}.\]
Thus, determining $\boldx$ is sufficient for computing the sought-for $\comp{G'}{H'}$. It remains to show that the matrix $\boldA$ is of full rank and is therefore invertible.

If $t<i$, we observe that $a_{t,i} = 0$ as we cannot use at least $i$ vertices of $H$ when we have fewer than $i$ vertices in the domain. For the diagonal elements with $t\in \{0,\dots,q\}$ we have that $a_{t,t}= N^t(V_t)=t!$ (note that $0!=1$). Hence, 
\begin{align*}
\boldA =
\begin{pmatrix}
0! & 0 &\cdots &0\\
\ast & 1! & \ddots &\vdots\\
\vdots & \ddots & \ddots & 0\\
\ast & \cdots &\ast & q!
\end{pmatrix}
\end{align*}
is a triangular matrix with non-zero diagonal entries, which completes the proof.
\end{proof}

\begin{lem}\label{lem:HomToComp2}
Let $H$ be a graph  without any size-$1$ components. Then
$\Hom{H} \le \Comp{H}$.
\end{lem}
\begin{proof}
The proof is by interpolation and is somewhat similar to the proof of Lemma~\ref{lem:HomToComp1}. Let $G$ be an input to \Hom{H}. We design a graph $G_t= G\oplus I_t$ as an input to the problem \Comp{H} by adding a set $I_t$ of $t$ disjoint new edges to the graph $G$.

We introduce some notation. Let $E^0(H)$ be the set of non-loop edges of $H$ and let $r= \abs{E^0(H)}$. Let $S^k(G)$ be the number of homomorphisms $\sigma$ from $G$ to $H$ that use exactly $k$ of the non-loop edges of $H$ (additionally, $\sigma$ might use any number of loops). Let $\{e_1, \dots, e_k\}$ be a set of $k$ arbitrary but fixed non-loop edges from $H$. We define $N^k(I_t)$ as the number of homomorphisms $\tau$ from $I_t$ to $H$ such that $\{e_1, \dots, e_k\}$ are amongst the edges used by $\tau$. Note that the particular choice of edges $\{e_1, \dots, e_k\}$ is not important when counting homomorphisms from an independent set of edges to $H$---$N^k(I_t)$ only depends on the numbers $k$ and $t$.

We observe that, for each compaction $\gamma\from V(G_t) \to V(H)$, the
restriction $\gamma\vert_{V(G)}$ uses some set $F\subseteq E^0(H)$ of non-loop
edges and does not use any other non-loop edges of $H$. Suppose that $F$ has
cardinality $\abs{F}=r-k$ for some $k\in \{0,\dots, r\}$. Then
$\gamma\vert_{V(I_t)}$ uses at least the remaining $k$ fixed non-loop edges of
$H$. As $H$ does not have any  size-$1$ components, this ensures at the same time that $\gamma$ is surjective.

Therefore, we obtain the following linear equation for a fixed $t\ge 0$:
\[
\underbrace{\comp{G_t}{H}}_{b_t} =\sum_{k=0}^r \underbrace{S^{r-k}(G)}_{x_k}\underbrace{N^k(I_t)}_{a_{t,k}}.
\]

As in the proof of Lemma~\ref{lem:HomToComp1}, we choose $t=0,\dots ,r$ to obtain a system of linear equations with
\begin{equation*}
\boldb=
\begin{pmatrix}
b_0\\
\vdots\\
b_r
\end{pmatrix}
\qquad
\boldA=
\begin{pmatrix}
a_{0,0} &\dots & a_{0,r}\\
\vdots & \ddots & \vdots\\
a_{r,0} &\dots & a_{r,r}\\
\end{pmatrix}
\qquad
\text{and}
\qquad
\boldx= \begin{pmatrix}
x_0\\
\vdots\\
x_r
\end{pmatrix}.
\end{equation*}
We can compute $\boldb$ using a \Comp{H} oracle. Further, 
\[\sum_{k=0}^r x_k = \sum_{k=0}^r S^{r-k}(G)= \sum_{k=0}^r S^k(G) = \hom{G}{H}.\]
Thus, determining $\boldx$ is sufficient for computing the sought-for number of homomorphisms $\hom{G}{H}$. 

Finally, we show that $\boldA$ is invertible. If $t<k$, we observe that $a_{t,k} = N^k(I_t) = 0$, as clearly it is impossible to use more than $t$ edges of $H$ when there are only $t$ edges in $I_t$. Further, for the diagonal elements it holds that for $t\in [r]$ we have $a_{t,t}= N^t(I_t)=2^tt!$ as there are $t!$ possibilities for assigning the edges in $I_t$ to the fixed set of $t$ edges of $H$ and there are $2^t$ vertex mappings for each such assignment of edges, also $N^0(I_0)=1$. Hence, 
\begin{align*}
\boldA =
\begin{pmatrix}
1 & 0 &\cdots &0\\
\ast & 2^11! & \ddots &\vdots\\
\vdots & \ddots & \ddots & 0\\
\ast & \cdots &\ast & 2^rr!
\end{pmatrix}
\end{align*}
is a triangular matrix with non-zero diagonal entries and is therefore invertible.
\end{proof}

\begin{thm}\label{thm:HomToComp}
Let $H$ be a graph. Then
$\Hom{H} \le \Comp{H}$.
\end{thm}
\begin{proof}
Let   $H'$ be the graph constructed from $H$ by removing all  size-$1$ components. By Lemma~\ref{lem:HomToComp1} we obtain
$\Comp{H'}\le \Comp{H}$.
Then Lemma~\ref{lem:HomToComp2} can be applied to the graph $H'$ and thus we obtain
$\Hom{H'}\le \Comp{H'}\le \Comp{H}$.
Finally, it follows from Theorem~\ref{thm:DGorig} that $\Hom{H'}\equiv \Hom{H}$, which gives
$\Hom{H} \equiv\Hom{H'}\le \Comp{H'}\le \Comp{H}$.
\end{proof}

Theorem~\ref{thm:HomToComp} shows that hardness results 
from Theorem~\ref{thm:DGorig} will
carry over from \Hom{H} to \Comp{H}. 
We also know some 
cases where \Comp{H} is tractable from Lemma~\ref{lem:LCompFP}. 
The complexity of \Comp{H} is still unresolved if
every component of~$H$ is a reflexive clique or an irreflexive biclique,
but some reflexive clique has size greater than~$2$, or some irreflexive biclique is not a star. 
This is the case described at length at the beginning of the section.  
Recall that the first step  is to specify  a technique (Theorem~\ref{thm:CompDecomp})
that lets us compute the number of compactions from a connected graph $G$ to a connected graph $H$ 
using a weighted sum of homomorphism counts,
say $\hom{G}{J_1},\ldots,\hom{G}{J_k}$. 
Towards this end,
 we introduce some definitions which we will use repeatedly in the remainder of this section.

\begin{defn}\label{defn:wGS}
A \emph{weighted graph set} is a tuple $(\calH,\lambda)$, where $\calH$ is a set of \emph{non-empty, pairwise non-isomorphic, connected} graphs and $\lambda$ is a function $\lambda\from \calH\to\Z$.
\end{defn}

\begin{defn}\label{defn:calS_H}
Let $H$ be a connected graph. By $Sub(H)$ we denote the set of non-empty, loop-hereditary, connected subgraphs of $H$. Let $\calS_H$ be a set which contains exactly one representative of each isomorphism class of the graphs in $Sub(H)$. Finally, for $H' \in \calS_H$, we define $\mu_H(H')$ to be the number of graphs in $Sub(H)$ that are isomorphic to $H'$.
\end{defn}

Note that for a connected graph $H$, we have $\mu_H(H)=1$. 

\begin{defn}\label{defn:lambda_H}
For each non-empty connected graph $H$, we define a weight function $\lambda_H$ which assigns an integer weight to each non-empty connected graph $J$. 
\begin{itemize}
\item 
If $J$ is not isomorphic to any graph in $\calS_H$,  then $\lambda_H(J)=0$. 
\item If   $J\cong H$,  then
$\lambda_H(J)=1$.
\item 
Finally, if $J$ is isomorphic to some graph in $\calS_H$ but $J\ncong H$, we define~$\lambda_H(J)$ inductively 
as follows.
\[\lambda_H(J)=-\sum_{\substack{H'\in \calS_H \\ \text{s.t. }H'\ncong H}} \mu_H(H')\lambda_{H'}(J).\]
\end{itemize}
Note that $\lambda_H$ is well-defined as all graphs $H'\in \calS_H$ with $H'\ncong H$ are smaller than $H$ either in the sense of having fewer vertices or in the sense of having the same number of vertices but fewer edges.
\end{defn}

The following theorem is the key to our approach for 
computing the number of compactions from a connected graph $G$ to a connected graph $H$ using a
weighted sum of homomorphism counts. In  the Appendix, we give an illustrative example where we verify the theorem for the case $H=K_{2,3}$ and 
we give the intuition behind the definitions. 
Here we go on to give the formal statement and proof.

\begin{thm}\label{thm:CompDecomp}
Let $H$ be a non-empty connected graph. Then for every non-empty, irreflexive and connected graph $G$ we have
$
\comp{G}{H} = \sum_{J\in \calS_H} \lambda_H(J)\hom{G}{J}$.
\end{thm}

\begin{proof}
Let $H_1, H_2, \dots$ be the set of non-empty connected graphs sorted by some fixed ordering that ensures that if $H_i$ is isomorphic to a subgraph of $H_j$, then $i\le j$. 
We verify the statement of the theorem by induction over the graph index with respect to this ordering. Let $G$ be non-empty, irreflexive and connected.

For the base case, $H_1$ is $K_1$, which is the graph with one vertex and no edges.
In this case, $\calS_{H_1}=\{K_1\}$ and $\lambda_{K_1}(K_1)=1$. Also 
\[\comp{G}{K_1} = \hom{G}{K_1}.\]
So the theorem holds in this case. 

Now assume that the statement holds for all graphs up to index $i$ and consider the graph $H_{i+1}$. For ease of notation we set $H=H_{i+1}$. We use the fact that every homomorphism from a connected graph $G$ to $H_{i+1}$ is a compaction onto some non-empty, loop-hereditary and connected subgraph of $H_{i+1}$ and vice versa. Thus, it holds that
\begin{align*}
\hom{G}{H} &= \sum_{H'\in \calS_{H}} \mu_H(H')\cdot\comp{G}{H'}\\
&= \comp{G}{H}+ \sum_{\substack{H'\in \calS_H \\ \text{s.t. }H'\ncong H}} \mu_H(H')\cdot \comp{G}{H'}.
\end{align*}
Thus, we can rearrange and use the induction hypothesis to obtain
\begin{align*}
\comp{G}{H} &= \hom{G}{H} - \sum_{\substack{H'\in \calS_H \\ \text{s.t. }H'\ncong H}} \mu_H(H')\cdot \comp{G}{H'}\\
&= \hom{G}{H} - \sum_{\substack{H'\in \calS_H \\ \text{s.t. }H'\ncong H}} \mu_H(H') \cdot \sum_{J\in \calS_{H'}} \lambda_{H'}(J)\hom{G}{J}.\\
\intertext{Then we change the order of summation and use that $\lambda_{H'}(J)=0$ if $J$ is not isomorphic to any graph in $\calS_{H'}$ to collect all coefficients that belong to a particular term $\hom{G}{J}$. We obtain}
\comp{G}{H}&=\hom{G}{H} - \sum_{\substack{J\in \calS_H \\ \text{s.t. }J\ncong H}} \Bigl(\sum_{\substack{H'\in \calS_H \\ \text{s.t. }H'\ncong H}}\mu_H(H')\lambda_{H'}(J)\Bigr)\hom{G}{J}\\
&= \sum_{J\in \calS_H} \lambda_H(J)\hom{G}{J}.
\end{align*}
\end{proof}

We remark that Theorem~\ref{thm:CompDecomp} can be generalised to graphs $H$ and $G$ with multiple connected components by looking at all subgraphs of $H$, rather than just at the connected ones. However, within this work, the version for connected graphs suffices.

Let $(\calH,\lambda)$ be a weighted graph set. The following parameterised problem is not natural in its own right, 
but it helps us to analyse the complexity of \cComp{H}:
\prob{\cGSHom{(\calH,\lambda)}.}{An irreflexive, connected graph $G$.}{$Z_{\calH,\lambda}(G)= \begin{cases}
0 &\quad \text{if $G$ is empty}\\
\sum_{J\in \calH} \lambda(J)\hom{G}{J}&\quad \text{otherwise.}
\end{cases}$}

\begin{cor}
\label{cor:Comp=GSHom}
Let $H$ be a non-empty connected graph. Then
$$\cComp{H} \equiv \cGSHom{(\calS_H,\lambda_H)}.$$
\end{cor}
\begin{proof}
The corollary follows directly from Theorem~\ref{thm:CompDecomp}.
\end{proof}

Corollary~\ref{cor:Comp=GSHom} gives us the desired connection between 
weighted graph sets and compactions.  We will use this later in the proof of Lemma~\ref{thm:simpleCompDichotomy}
to establish the $\numP$-hardness of $\cComp{H}$ when $H$ is either a reflexive clique of size at least~$3$
or an irreflexive biclique that is not a star.

Our next goal is to prove Theorem~\ref{thm:CompInterpol1},
which states that,  for certain weighted graph sets $(\calH,\lambda)$, 
determining $Z_{\calH,\lambda}(G)$ is at least as hard as computing 
$\hom{G}{J}$ for some graph $J$ from the set $\calH$ with $\lambda(J)\neq 0$. 
To this end, we first introduce two lemmas that   help us to distinguish between 
different graphs $J$ in the interpolation 
that we will later use to prove Theorem~\ref{thm:CompInterpol1}.

For the following lemmas, we introduce some new notation. For a graph $G$ with distinguished vertex $v\in V(G)$ and a graph $H$ with distinguished vertex $w\in V(H)$, the quantity
$\hom{(G,v)}{(H,w)}$
denotes the number of homomorphisms $h$ from $G$ to $H$ with $h(v)=w$. Analogously, $\inj{(G,v)}{(H,w)}$ denotes the number of injective homomorphisms $h$ from $G$ to $H$ with $h(v)=w$. 
If there exists an isomorphism from $G$ to $H$ that maps $v$ onto $w$, we write $(G,v) \cong (H,w)$, otherwise we write $(G,v) \ncong (H,w)$. In the following lemma, we show that for two such target entities $(H_1 ,w_1)$ and $(H_2, w_2)$ that are non-isomorphic, there exists an input which separates them. To this end, we use an argument very similar to that presented in~\cite[Lemma 3.6]{Gobel2016} and in the textbook by Hell and Ne\v{s}et\v{r}il~\cite[Theorem 2.11]{HellNesetrilBook}, which goes back to the works of Lov\'asz~\cite[Theorem 3.6]{Lovasz1967}. 

\begin{lem}\label{lem:Lovasz1}
Let $H_1$ and $H_2$ be connected graphs with distinguished vertices $w_1\in V(H_1)$ and $w_2\in V(H_2)$ such that $(H_1 ,w_1) \ncong (H_2, w_2)$.  Suppose that one of the following cases holds:
\begin{itemize}
\settowidth{\templength}{Case}
\setlength{\itemindent}{\templength}
\item[Case 1.] $H_1$ and $H_2$ are reflexive graphs.
\item[Case 2.] $H_1$ and $H_2$ are irreflexive bipartite graphs, each of which contains at least one edge.
\end{itemize}
Then 
\begin{enumerate}
\renewcommand\labelenumi{\roman{enumi})}
\item There exists a connected irreflexive graph $G$ with distinguished vertex $v\in V(G)$ for which
$
\hom{(G,v)}{(H_1,w_1)} \neq \hom{(G,v)}{(H_2,w_2)}$.
\item In Case 2 we can assume that $G$ contains at least one edge and is bipartite.
\end{enumerate}
\end{lem}
\begin{proof}

In order to shorten the proof, we define some notation that depends on which case holds. In Case 1, we say that a tuple $(G,v)$ consisting of a graph $G$ with distinguished vertex $v$ is \emph{relevant} if $G$ is connected and reflexive. In Case 2, we say that it is relevant if $G$ is connected, irreflexive and bipartite and contains at least one edge. We start with a claim that applies in either case.

\medskip
\noindent {\bf Claim: There exists a relevant $(G,v)$ such that
$$\hom{(G,v)}{(H_1,w_1)} \neq \hom{(G,v)}{(H_2,w_2)}.$$}
\medskip

\noindent {\bf Proof of the claim:}\quad
To prove the claim, assume for a contradiction that for all relevant $(G,v)$ we have
\begin{equation}\label{equ:Lovasz1}
\hom{(G,v)}{(H_1,w_1)} = \hom{(G,v)}{(H_2,w_2)}.
\end{equation}
The contradiction will follow from the following subclaim:

\medskip
\noindent{\bf Subclaim: For every relevant $(G,v)$,}
$\inj{(G,v)}{(H_1,w_1)} = \inj{(G,v)}{(H_2,w_2)}${\bf.}

\medskip 

\noindent {\bf Proof of the subclaim:}\quad
The proof of the subclaim is by induction on the number of vertices of $G$. For the base case of the induction we treat the two cases separately. 

In Case 1, the base case of the induction is $\abs{V(G)}=1$.
The relevant $(G,v)$ is the graph consisting of the single (looped) vertex~$v$.
For every reflexive graph $H$ and vertex $w\in V(H)$ we have that $\hom{(G,v)}{(H,w)} = \inj{(G,v)}{(H,w)}$. 
Therefore, (\ref{equ:Lovasz1}) implies that the subclaim is true for this $(G,v)$.
 
In Case 2, the base case of the induction is $\abs{V(G)}=2$.
(There are no relevant $(G,v)$ with $\abs{V(G)}<2$ since $G$ has to contain an edge.) 
Consider a relevant $(H,w)$. Since $H$ is irreflexive
and the two vertices of $G$ are connected by an edge (so cannot be mapped by a homomorphism to
the same vertex of $H$) we have  $\hom{(G,v)}{(H,w)} = \inj{(G,v)}{(H,w)}$.
Once again, (\ref{equ:Lovasz1}) implies that the subclaim is true for this $(G,v)$.

\begin{figure}[t]\centering
\def\scaleFactor{.9}
\begin{minipage}{.3\linewidth}\centering
\begin{tikzpicture}[scale=\scaleFactor, baseline=\scaleFactor*(2cm-.2\baselineskip)]

			\coordinate (v1)  at (0, 4);
			\node at (v1) [left = 1mm of v1] {$v_1$};
			\coordinate (v4)  at (0, 0);
			\node at (v4) [right = 1mm of v4] {$v_4$};
			\coordinate (v3)  at (1.5, 2);
			\node at (v3) [right = 1mm of v3] {$v_3$};
			\coordinate (v2)  at (3, 4);
			\node at (v2) [right = 1mm of v2] {$v_2$};
			\coordinate (v5)  at (3, 0);
			\node at (v5) [left = 1mm of v5] {$v_5$};
			
			\fill (v1) circle[radius=3pt];
			\fill (v2) circle[radius=3pt];
			\fill (v3) circle[radius=3pt];
			\fill (v4) circle[radius=3pt];
			\fill (v5) circle[radius=3pt];
	
			\draw (v1) -- (v2) -- (v3) -- cycle;
			\draw (v3) -- (v4);
			\draw (v3) -- (v5);
			\addvmargin{7mm}
\end{tikzpicture}\\
$G$
\end{minipage}
\begin{minipage}{.3\linewidth}\centering
\begin{tikzpicture}[scale=\scaleFactor, baseline=\scaleFactor*(2cm-.2\baselineskip)]

			\coordinate (v1)  at (0, 4);
			\node at (v1) [left = 1mm of v1] {$v_1$};
			\coordinate (v4)  at (0, 0);
			\node at (v4) [right = 1mm of v4] {$v_4$};
			\coordinate (v3)  at (1.5, 2);
			\node at (v3) [right = 1mm of v3] {$v_3$};
			\coordinate (v2)  at (3, 4);
			\node at (v2) [right = 1mm of v2] {$v_2$};
			\coordinate (v5)  at (3, 0);
			\node at (v5) [left = 1mm of v5] {$v_5$};
			
			\fill (v1) circle[radius=3pt];
			\fill (v2) circle[radius=3pt];
			\fill (v3) circle[radius=3pt];
			\fill (v4) circle[radius=3pt];
			\fill (v5) circle[radius=3pt];
	
			\draw (v1) -- (v2) -- (v3) -- cycle;
			\draw (v3) -- (v4);
			\draw (v3) -- (v5);
			
			{
			\coordinate (ecenter)  at (.75, 3);
			\begin{scope}[rotate around={-51:(ecenter)}]
    			\draw (ecenter) ellipse (2cm and 1cm);
			\end{scope}
			}
			
			{
			\coordinate (ecenter)  at (1.5, 0);
    		\draw (ecenter) ellipse (2cm and 1cm);
			}
			
			\draw (v2) circle[radius=.75cm];
			\addvmargin{2mm}
\end{tikzpicture}\\
$\theta$
\end{minipage}
\begin{minipage}{.3\linewidth}\centering
\begin{tikzpicture}[scale=\scaleFactor, baseline=\scaleFactor*(2cm-.2\baselineskip), every loop/.style={min distance=15mm,looseness=10}]

			\coordinate (v1)  at (1.5, 4);
			\node at (v1) [left = 1mm of v1] {$\{v_2\}$};
			\coordinate (v3)  at (1.5, 2);
			\node at (v3) [left = 1mm of v3] {$\{v_1,v_3\}$};
			\coordinate (v5)  at (1.5, 0);
			\node at (v5) [left = 1mm of v5] {$\{v_4,v_5\}$};
			
			\fill (v1) circle[radius=3pt];
			\fill (v3) circle[radius=3pt];
			\fill (v5) circle[radius=3pt];
	
			\path[-] (v3) edge  [in= 30,out= -30,loop] node {} ();
			
			\draw (v1) -- (v3) -- (v5);
			\addvmargin{7mm}
\end{tikzpicture}\\
$G|_\theta$
\end{minipage}
\caption{Graph $G$ and the corresponding quotient graph $G\vert_\theta$ for $\theta=\{\{v_2\}, \{v_1,v_3\},\{v_4,v_5\}\}$.}
\label{fig:Quotient}
\end{figure}
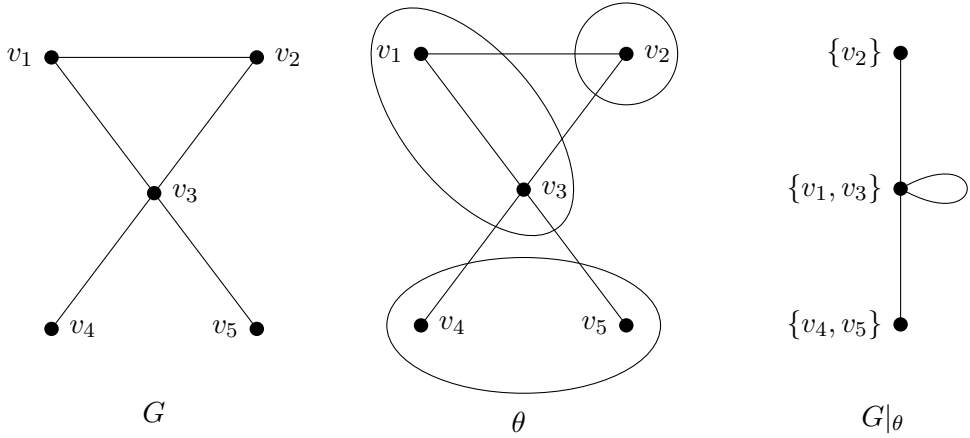
 
For the inductive step, suppose that the subclaim holds for all relevant $(G,v)$ 
in which $G$ has up to $k-1$ vertices. Consider a relevant $(G,v)$ with $\abs{V(G)}=k$. Let $\Theta$ be the set of partitions of $V(G)$ --- that is, each $\theta\in \Theta$ is a set $\{U_1,\dots, U_j\}$ for some integer $j$ such that the elements of $\theta$ are non-empty and pairwise disjoint subsets of $V(G)$ with $\bigcup_{i=1}^j U_i = V(G)$. For $\theta\in \Theta$ with $\theta  =  \{U_1,\dots, U_j\}$, by $G|_\theta$ we denote the corresponding \emph{quotient graph}, i.e.\ let $G|_\theta$ be the graph with vertices $\{U_1,\dots, U_j\}$ that has an edge $\{U_i,U_{i'}\}$ if and only if there exist $v\in U_i$ and $u\in U_{i'}$ with $\{v,u\}\in E(G)$. Therefore, $G|_\theta$ might have loops but no multi-edges, see Figure~\ref{fig:Quotient}. Let $v_\theta$ denote the vertex of $G|_\theta$ which corresponds to the equivalence class of $\theta$ that contains the distinguished vertex $v$. Finally, let $\tau$ denote the partition of $V(G)$ into singletons. Then for every relevant $(H,w)$ it holds that
\begin{align}
\hom{(G,v)}{(H,w)} &= \sum_{\theta\in \Theta} \inj{(G|_\theta,v_\theta)}{(H,w)}\nonumber\\
&=\inj{(G|_\tau,v_\tau)}{(H,w)} +\sum_{\theta\in \Theta\setminus\{\tau\}} \inj{(G|_\theta,v_\theta)}{(H,w)}\nonumber\\
&=\inj{(G,v)}{(H,w)} +\sum_{\theta\in \Theta\setminus\{\tau\}} \inj{(G|_\theta,v_\theta)}{(H,w)},\label{equ:Lovasz4}
\end{align}
where the third equality follows as $G|_\tau = G$.

Now we show that only relevant tuples $(G|_\theta,v_\theta)$ actually contribute to the sum in~(\ref{equ:Lovasz4}). First, note that since $G$ is connected, so is $G|_\theta$.

In Case 1, every quotient graph $G|_\theta$ is reflexive. Therefore, for every $\theta\in \Theta\setminus\{\tau\}$, the tuple $(G|_\theta,v_\theta)$ is relevant.

In Case 2, $H$ is an irreflexive bipartite graph with at least one edge. Therefore, we have $\inj{(G|_\theta,v_\theta)}{(H,w)}>0$ only if $G|_\theta$ is an irreflexive bipartite graph and also, $\theta$ is a proper vertex-colouring of $G$, i.e.\ every part of $\theta$ is an independent set. For such a partition $\theta$, $G\vert_\theta$ has at least one edge if $G$ does. We have now shown that only relevant tuples $(G|_\theta,v_\theta)$ contribute to the sum in~(\ref{equ:Lovasz4}).

Therefore, let $\Gamma$ be the set of all partitions $\theta$ of $V(G)$ such that $(G|_\theta,v_\theta)$ is relevant. Then, we can rephrase~(\ref{equ:Lovasz4}) as follows.
\begin{equation}\label{equ:Lovasz3}
\hom{(G,v)}{(H,w)} =\inj{(G,v)}{(H,w)} +\sum_{\theta\in \Gamma\setminus\{\tau\}} \inj{(G|_\theta,v_\theta)}{(H,w)}.
\end{equation}
To prove the subclaim, we can set $(H,w)$ in~(\ref{equ:Lovasz3}) to be $(H_1,w_1)$. Similarly, we can set it to be $(H_2,w_2)$. Then, we can use the induction hypothesis, the subclaim, on all tuples $(G|_\theta,v_\theta)$ in the sum as all these tuples are relevant and the partitions $\theta \in \Gamma\setminus\{\tau\}$ have strictly fewer than $k$ parts. Applying~(\ref{equ:Lovasz1}), we obtain
\[
\inj{(G,v)}{(H_1,w_1)} = \inj{(G,v)}{(H_2,w_2)},
\]
which completes the induction and the proof of the subclaim. 
{\bf (End of the proof of the subclaim.)}
\medskip
 
 We show next how to use the subclaim  to derive a contradiction. In particular, in the subclaim we can set $(G,v)$ to be either $(H_1,w_1)$ or $(H_2,w_2)$. This implies $(H_1 ,w_1) \cong (H_2, w_2)$, which gives the desired contradiction. Thus, we have shown contrary to~(\ref{equ:Lovasz1}) that there exists a relevant $(G,v)$ with
\[
\hom{(G,v)}{(H_1,w_1)} \neq \hom{(G,v)}{(H_2,w_2)}
\]
and therefore we have proved the claim.
{\bf (End of the proof of the claim.)}
\medskip 

In Case 2, the claim is identical to the statement of the lemma. However, in Case 1 a relevant tuple $(G,v)$ contains a reflexive graph $G$, whereas for the statement of the lemma, $G$ has to be irreflexive. This is easily fixed as we can set $G^0$ to be the graph constructed from $G$ by removing all loops. Using the fact that $H_1$ and $H_2$ are reflexive, we obtain for $i=1$ and $i=2$ that
\begin{equation*}
\hom{(G^0,v)}{(H_i,w_i)} = \hom{(G,v)}{(H_i,w_i)}.
\end{equation*}
Hence, the choice $(G^0,v)$ has all the desired properties.
\end{proof}

In the following lemma, we generalise the pairwise property from Lemma~\ref{lem:Lovasz1}. The result and the proof are adapted versions of~\cite[Lemma 6]{Galanis2016}. For ease of notation let $\binom{[k]}{2}$ denote the set of all pairs $\{i,j\}$ with $i,j\in [k]$ and $i\neq j$.

\begin{lem}\label{lem:Lovasz2}
Let $H_1, \dots, H_k$ be connected graphs with distinguished vertices $w_1, \dots, w_k$ where $w_i\in V(H_i)$ for all $i\in [k]$ and, for every pair $\{i,j\} \in \binom{[k]}{2}$, we have $(H_i ,w_i)\ \ncong (H_j, w_j)$.  Suppose that one of the following cases holds:
\begin{itemize}
\settowidth{\templength}{Case}
\setlength{\itemindent}{\templength}
\item[Case 1.] $\forall i\in [k]$, $H_i$ is a reflexive graph.
\item[Case 2.] $\forall i\in [k]$, $H_i$ is an irreflexive bipartite graph that contains at least one edge.
\end{itemize}
Then 
\begin{enumerate}
\renewcommand\labelenumi{\roman{enumi})}
\item There exists a connected irreflexive graph $G$ with a distinguished vertex $v\in V(G)$ such that, for every $\{i,j\} \in \binom{[k]}{2}$, it holds that
$
\hom{(G,v)}{(H_i,w_i)} \neq \hom{(G,v)}{(H_j,w_j)}$.
\item In Case 2 we can assume that $G$ contains at least one edge and is bipartite.
\end{enumerate}
\end{lem}
\begin{proof}
Again, we use the notion of relevant tuples but slightly modify the definition from the one given in the proof of Lemma~\ref{lem:Lovasz1}. A tuple $(G,v)$ is called relevant if $G$ is a connected \emph{irreflexive} graph and, in Case 2, if additionally $G$ contains at least one edge and is bipartite. We show that there exists a relevant $(G,v)$ such that for every $\{i,j\} \in \binom{[k]}{2}$ we have
\[
\hom{(G,v)}{(H_i,w_i)} \neq \hom{(G,v)}{(H_j,w_j)}.
\]

We use induction on $k$, which is the number of graphs $H_1,\dots, H_k$. The base case for $k=2$ is covered by Lemma~\ref{lem:Lovasz1}. Now let us assume that the statement holds for $k-1$ and the inductive step is for $k$. By the inductive hypothesis there exists a relevant $(G,v)$ such that without loss of generality (possibly by renaming the graphs $H_1,\dots, H_k$)
\[\hom{(G,v)}{(H_2,w_2)} >\dots > \hom{(G,v)}{(H_k,w_k)}.\]
Let $i^*\in[k]\setminus\{1\}$ be an index with 
\[\hom{(G,v)}{(H_1,w_1)} = \hom{(G,v)}{(H_{i^*},w_{i^*})}.\]
If no such index exists, we can simply choose the graph $G$ which then fulfils the statement of the lemma. Using the base case, there exists a relevant $(G',v')$ such that 
\[\hom{(G',v')}{(H_1,w_1)} > \hom{(G',v')}{(H_{i^*},w_{i^*})},\]
possibly renaming $(H_1,w_1)$ and $(H_{i^*},w_{i^*})$. 
Let $i \in [k]$. 

First, we show that for all $i\in[k]$ we have $\hom{(G',v')}{(H_i,w_i)} \ge 1$. This is clearly true for Case 1, where $w_i$ has a loop. In this case, we can always map all vertices of $G'$ to the single vertex $w_i$. 

In Case 2, as $H_i$ is connected and contains at least one edge, there is some $w\in V(H_i)$ such that $\{w,w_i\}\in E(H_i)$. Since $(G',v')$ is relevant, $G'$ is connected and bipartite and has at least one edge. Let $\{A,B\}$ be a partition of $V(G')$ such that $v'\in A$ and $A$ and $B$ are independent sets of $G$.  There is a homomorphism $h$ from $G'$ to $H_i$ with $h(v')=w_i$ which maps all vertices in $A$ to $w_i$ and all vertices in $B$ to $w$. 

Therefore, in both cases we have shown that for all $i\in[k]$ we have $\hom{(G',v')}{(H_i,w_i)} \ge 1$. 

\begin{figure}[h!]
\centering
\def\scaleFactor{.6}
\begin{tikzpicture}[scale=\scaleFactor, baseline=\scaleFactor*(2cm-.2\baselineskip)]

			\coordinate (ccenter)  at (-2.5, 0);
			\node at (ccenter) {$G'$};
			\coordinate[label=right:{$v^*=v'=v$}](eleftend) at (1.5, 0);	
			\coordinate[label=$1$] (etag1)  at (12, 3.5);
			\coordinate[label=$2$] (etag2)  at (12, 1.5);	
			\coordinate (dot1)  at (12, -.25);
			\coordinate (dot2)  at (12, -.75);
			\coordinate (dot3)  at (12, -1.25);
			\coordinate[label=$t$] (etag3)  at (12, -4);
			
			\fill (eleftend) circle[radius=3pt];
			\fill (dot1) circle[radius=2pt];
			\fill (dot2) circle[radius=2pt];
			\fill (dot3) circle[radius=2pt];
			
			\draw (ccenter) circle (4cm);
			
			\begin{scope}[rotate around={20:(eleftend)}]
				\node at (10.75,0) {$G$};
    			\coordinate (ecenter)  at (6.5, 0);
    			\draw (ecenter) ellipse (5cm and 1.8cm);
			\end{scope}
			\begin{scope}[rotate around={10:(eleftend)}]
				\node at (10.75,0) {$G$};
    			\coordinate (ecenter)  at (6.5, 0);
    			\draw (ecenter) ellipse (5cm and 1.8cm);
			\end{scope}
			\begin{scope}[rotate around={-20:(eleftend)}]
				\node at (10.75,0) {$G$};
    			\coordinate (ecenter)  at (6.5, 0);
    			\draw (ecenter) ellipse (5cm and 1.8cm);
			\end{scope}

			\addvmargin{3mm}
\end{tikzpicture}
\caption{$(G^*, v^*)$.}
\label{fig:LovaszG*}
\end{figure}

For a yet to be determined number $t$ we construct a graph $G^*$ from $(G,v)$ and $(G',v')$ by taking the graph $G'$ and $t$ copies of $G$ and identifying the vertex $v'$ with the $t$ copies of $v$ and call the resulting vertex $v^*$, cf. Figure~\ref{fig:LovaszG*}. Note that from the fact that $(G,v)$ and $(G',v')$ are relevant, it is straightforward to show that $(G^*,v^*)$ is relevant as well. Then, for any graph $H$ and $w\in V(H)$ it holds that
\[\hom{(G^*,v^*)}{(H,w)} = \hom{(G',v')}{(H,w)}\cdot\hom{(G,v)}{(H,w)}^t.\]
The goal is to choose $t$ sufficiently large to achieve
\begin{align*}
\hom{(G^*,v^*)}{(H_2,w_2)} > \dots &> \hom{(G^*,v^*)}{(H_{{i^*}-1},w_{{i^*}-1})}\\ 
&> \hom{(G^*,v^*)}{(H_1,w_1)}\\
&> \hom{(G^*,v^*)}{(H_{i^*},w_{i^*})} > \dots > \hom{(G^*,v^*)}{(H_k,w_k)}.
\end{align*}
Accordingly, we define a permutation $\sigma$ of the indices $\{1,\dots, k\}$ that inserts index $1$ between position ${i^*}-1$ and $i^*$. The domain of $\sigma$ corresponds to the new indices to which we assign the former indices. To avoid confusion, we give the function table in Table~\ref{tab:sigma}
\begin{table}[h]
\centering
\def\arraystretch{1.5}
\begin{tabular}{c|ccccccc}
$i$ & $1$ & $\cdots$ & $i^*-2$ & $i^*-1$ & $i^*$ & $\cdots$ & $k$\\
\hline
$\sigma(i)$ & $2$ & $\cdots$ & $i^*-1$ & $1$ & $i^*$ & $\cdots$ & $k$
\end{tabular}
\caption{Function table of $\sigma$.}
\label{tab:sigma}
\end{table}

Formally,
\[\sigma(i)=
\begin{cases}
i+1 \quad\text{if }i\le {i^*}-2\\
1 \quad\text{if }i= {i^*}-1\\
i\quad\text{otherwise.}
\end{cases}
\]
Let $M= \hom{(G,v)}{(H_2,w_2)}$. As $\hom{(G',v')}{(H_j,w_j)} \ge 1$ for all $j \in [k]$, it is well-defined to set
\[C= \max_{j\in [k]\setminus\{{i^*}-1\}} \frac{\hom{(G',v')}{(H_{\sigma(j+1)},w_{\sigma(j+1)})}}{\hom{(G',v')}{(H_{\sigma(j)},w_{\sigma(j)})}}\]
and $t= \lceil CM \rceil$. Let $G^*$ be as defined above. For ease of notation, for $j\in[k-1]$, we set
\[\xi(j)=\frac{\hom{(G^*,v^*)}{(H_{\sigma(j)},w_{\sigma(j)}}}{\hom{(G^*,v^*)}{(H_{\sigma(j+1)},w_{\sigma(j+1)}}}.\]
We want to show $\xi(j)>1$ for all $j\in[k-1]$ to complete the proof.

For $j=i^*-1$ we obtain
\begin{align*}
\xi(j) &=\frac{\hom{(G^*,v^*)}{(H_{\sigma({i^*}-1)},w_{\sigma({i^*}-1)})}}{\hom{(G^*,v^*)}{(H_{\sigma({i^*})},w_{\sigma({i^*})})}}\\
&=\frac{\hom{(G^*,v^*)}{(H_1,w_1)}}{\hom{(G^*,v^*)}{(H_{i^*},w_{i^*})}}\\
&=\frac{\hom{(G',v')}{(H_1,w_1)}}{\hom{(G',v')}{(H_{i^*},w_{i^*})}}> 1.
\end{align*} 

For $j\in[k-1]\setminus\{{i^*}-1\}$ we have
\begin{align*}
\xi(j) &=\frac{\hom{(G^*,v^*)}{(H_{\sigma(j)},w_{\sigma(j)})}}{\hom{(G^*,v^*)}{(H_{\sigma(j+1)},w_{\sigma(j+1)})}}\\ &= \frac{\hom{(G',v')}{(H_{\sigma(j)},w_{\sigma(j)})}\cdot\hom{(G,v)}{(H_{\sigma(j)},w_{\sigma(j)})}^t}{\hom{(G',v')}{(H_{\sigma(j+1)},w_{\sigma(j+1)})}\cdot\hom{(G,v)}{(H_{\sigma(j+1)},w_{\sigma(j+1)})}^t}\\
&\ge \frac{1}{C}\left(\frac{\hom{(G,v)}{(H_{\sigma(j)},w_{\sigma(j)})}}{\hom{(G,v)}{(H_{\sigma(j+1)},w_{\sigma(j+1)})}}\right)^t.
\end{align*}
Since $\hom{(G,v)}{(H_{\sigma(j)},w_{\sigma(j)})} \ge 1 + \hom{(G,v)}{(H_{\sigma(j+1)},w_{\sigma(j+1)})}$ for $j\in[k-1]\setminus\{{i^*}-1\}$ we have
\[
\xi(j) \ge\frac{1}{C}\left(1+\frac{1}{\hom{(G,v)}{(H_{\sigma(j+1)},w_{\sigma(j+1)})}}\right)^t.
\]
Using $(1+x)^t\ge 1+tx >tx$ for $t\ge 1$, $x\ge 0$ we obtain
\[
\xi(j) >\frac{t}{C\cdot\hom{(G,v)}{(H_{\sigma(j+1)},w_{\sigma(j+1)})}}.
\]
Finally, we use that for all $j\in[k-1]\setminus\{{i^*}-1\}$ we have \[\hom{(G,v)}{(H_2,w_2)}> \hom{(G,v)}{(H_{\sigma(j+1)},w_{\sigma(j+1)})}\] and conclude
\[
\xi(j)> \frac{t}{C\cdot\hom{(G,v)}{(H_2,w_2)}}\ge \frac{t}{CM} \ge 1.
\]
Thus, we have shown $\xi(j) >1$ as required, which completes the proof. 
\end{proof}

In the following theorem, we use the separating instances that we obtain from Lemma~\ref{lem:Lovasz2} for interpolation-based reductions to \cGSHom{(\calH, \lambda)}.
\begin{thm}\label{thm:CompInterpol1}
Let $(\calH,\lambda)$ be a weighted graph set for which one of two cases holds:
\begin{itemize}
\settowidth{\templength}{Case}
\setlength{\itemindent}{\templength}
\item[Case 1.] All graphs in $\calH$ are reflexive.
\item[Case 2.] All graphs in $\calH$ are irreflexive and bipartite.
\end{itemize}
Then, for all $H \in \calH$ with $\lambda(H)\neq 0$, it holds that
$\cHom{H} \le \cGSHom{(\calH,\lambda)}$.
\end{thm}
\begin{proof}
If, in Case 2, $\calH$ contains a graph without edges, i.e.\ a single-vertex graph $K_1$, let $(\calH', \lambda')$ be a weighted graph set constructed from $(\calH, \lambda)$ by removing the $K_1$ and its corresponding weight $\lambda(K_1)$. As $\Hom{K_1}$ is in $\FP$ we have 
\[\cGSHom{(\calH',\lambda')} \le \cGSHom{(\calH,\lambda)}\]
and 
\[\cHom{K_1} \le \cGSHom{(\calH,\lambda)}.\]
Therefore, for the remainder of this proof, we assume that every graph in $\calH$ contains at least one edge. Let $\calH^{\neq0}=\{H_1,\dots,H_k\}$ be the set of graphs in $\calH$ that are assigned non-zero weights by $\lambda$. Note that all graphs in $\calH^{\neq0}$ are pairwise non-isomorphic, connected and non-empty by definition of a weighted graph set. Thus, for every pair $\{i,j\} \in \binom{[k]}{2}$ and every $w_i\in V(H_i)$, $w_j\in V(H_j)$ we have $(H_i ,w_i)\ \ncong (H_j, w_j)$. 

Now, for each graph $H_i$ we collect the vertices which are in the same orbit of the automorphism group of $H_i$. Formally, for each $i\in [k]$ and $w\in V(H_i)$, let $[w]$ be the orbit of $w$, i.e.\ the set of vertices $w'$ such that $(H_i,w')\cong (H_i,w)$. Let $W$ be a set which contains exactly one representative from each such orbit. Further, for each $i\in[k]$ set $W_i= W\cap V(H_i)$. Then, for each $w,w'\in W_i$ with $w'\neq w$, we have $(H_i ,w)\ \ncong (H_i, w')$. 

Let $k'=\sum_{i=1}^k \abs{W_i}$ and let $(H'_1,w'_1), \dots, (H'_{k'},w'_{k'})$ be an enumeration of $\{(H_i,w_i) : i\in[k],\ w_i\in W_i\}$. Then we can apply Lemma~\ref{lem:Lovasz2} to $(H'_1,w'_1), \dots, (H'_{k'},w'_{k'})$ to obtain a connected irreflexive graph $J$ with distinguished $u\in V(J)$ such that for every $i,j\in [k]$ and for all $w_i\in W_i$, $w_j\in W_j$ we have $\hom{(J,u)}{(H_i,w_i)} \neq \hom{(J,u)}{(H_j,w_j)}$.

Let $i\in [k]$ and suppose that $H_i\in \calH$ and $\lambda(H_i) \neq 0$. Let $G$ be a non-empty graph which is an input to the problem \cHom{H_i}. Let $v$ be an arbitrary vertex of $G$. We use the same construction as in Figure~\ref{fig:LovaszG*} to design a graph $G_t$ as input to the problem \cGSHom{(\calH, \lambda)} by taking $t$ copies of $J$ as well as the graph $G$ and identifying the $t$ copies of vertex $u$ with the vertex $v\in V(G)$. As both $G$ and $J$ are connected, $G_t$ is as well.
Then, using an oracle for \cGSHom{(\calH, \lambda)}, we can compute $Z_{\calH,\lambda}(G_t)$ with
\begin{align}
Z_{\calH,\lambda}(G_t)&= \sum_{H\in \calH} \lambda(H)\hom{G_t}{H}\nonumber\\
&= \sum_{i\in [k]} \lambda(H_i)\hom{G_t}{H_i} \nonumber\\
&= \sum_{i\in [k]} \lambda(H_i)\sum_{w\in V(H_i)}\hom{(G,v)}{(H_i,w)}\cdot\hom{(J,u)}{(H_i,w)}^t \label{equ:Inter1}
\end{align}
Now we collect the terms which belong to vertices in the same orbit. To this end, for $w\in W$ and $i\in[k]$ such that $w\in V(H_i)$, we define $\lambda_w= \abs{[w]}\cdot \lambda(H_i)$, $N_w(G)= \hom{(G,v)}{(H_i,w)}$ and $N_w(J)=\hom{(J,u)}{(H_i,w)}$. Let $W=\{w_0,\dots,w_r\}$. Then, continuing from Equation~(\ref{equ:Inter1}):
\begin{align*}
Z_{\calH,\lambda}(G_t)&= \sum_{i\in [k]} \lambda(H_i)\sum_{w\in V(H_i)}\hom{(G,v)}{(H_i,w)}\cdot\hom{(J,u)}{(H_i,w)}^t\\
&=\sum_{w\in W} \lambda_w N_w(G)N_w(J)^t.
\end{align*}

By choosing $r+1$ different values for the parameter $t$ --- here it is sufficient to choose $t=0,\dots ,r$ --- we obtain a system of linear equations $\boldb= \boldA\boldx$ as follows:  
\begin{equation*}
\boldb=
\begin{pmatrix}
Z_{\calH,\lambda}(G_0)\\
\vdots\\
Z_{\calH,\lambda}(G_r)
\end{pmatrix}
\quad
\boldA=
\begin{pmatrix}
\lambda_{w_0}N_{w_0}(J)^0 &\dots & \lambda_{w_r}N_{w_r}(J)^0\\
\vdots & \ddots & \vdots\\
\lambda_{w_0}N_{w_0}(J)^r &\dots & \lambda_{w_r}N_{w_r}(J)^r\\
\end{pmatrix}
\quad
\text{and}
\quad
\boldx= \begin{pmatrix}
N_{w_0}(G)\\
\vdots\\
N_{w_r}(G)
\end{pmatrix}
\end{equation*}
The vector $\boldb$ can be computed using $r+1$ \cGSHom{(\calH, \lambda)} oracle calls. Then
\[\hom{G}{H_i} = \sum_{w \in W_i}\abs{[w]} N_{w}(G).\]
Thus, determining $x$ is sufficient for computing the sought-for $\hom{G}{H_i}$. It remains to show that the matrix $\boldA\in \Z^{(r+1)\times (r+1)}$ is of full rank and therefore invertible. This can be easily seen by dividing each column by its first entry. The division is well-defined as for $t\in \{0\dots, r\}$ we have $\lambda_{w_t}\neq0$ by definition of $\calH^{\neq0}$. The columns of the resulting matrix are pairwise different by the choice of $(J,u)$ and as a consequence the resulting matrix is a Vandermonde matrix and therefore invertible.
\end{proof}

Next, we give a short technical lemma which follows from Definition~\ref{defn:lambda_H} and is used in  Lemma~\ref{thm:simpleCompDichotomy} to show that Theorem~\ref{thm:CompInterpol1} gives hardness results for \cComp{H}.
\begin{lem}
\label{lem:H-}
Let $H$ be a connected graph with at least one non-loop edge. Let $H^-$ be the graph obtained 
from~$H$ by deleting exactly one non-loop edge   (but keeping all vertices). If $H^-$ is connected,  then
$\lambda_H(H^-) \neq 0$.
\end{lem}
\begin{proof}
As $H^-$ is non-empty and connected, it is a valid input to $\lambda_H$ and from the definition of $\lambda_H$ (Definition~\ref{defn:lambda_H}) we obtain
\begin{align*}
\lambda_H(H^-)&= -\sum_{\substack{H'\in \calS_H \\ \text{s.t. }H'\ncong H}} \mu_H(H')\lambda_{H'}(H^-).\\
\intertext{Consider a graph $H'\in \calS_H$ with $H'\ncong H$ and $H'\ncong H^-$. $H'$ is a non-empty loop-hereditary connected subgraph of $H$ and not isomorphic to $H$ or $H^-$. Note that $H^-$ is not isomorphic to any graph in $\calS_{H'}$ which gives $\lambda_{H'}(H^-) = 0$. Furthermore, $\mu_H(H^-) \ge 1$. Thus, we proceed} 
\lambda_H(H^-)&= -\mu_H(H^-)\lambda_{H^-}(H^-)\\
&\le -1.
\end{align*}
\end{proof}

We now have most of the tools at hand to classify the complexity of \Comp{H}.
Tractability results come from Lemma~\ref{lem:LCompFP}. If $H$ has a component that
is not a reflexive clique or an irreflexive biclique  then hardness will be lifted from Dyer and Greenhill's Theorem~\ref{thm:DGorig}
via Theorem~\ref{thm:HomToComp}. The most difficult case is when all components of $H$
are reflexive cliques or irreflexive bicliques, but some component is not an irreflexive 
star or a reflexive clique of size at most~$2$.

If $H$ is connected then hardness will come from the following lemma,  
whose
proof builds on the weighted graph set technology (Corollary~\ref{cor:Comp=GSHom})
using Theorem~\ref{thm:CompInterpol1} and Lemma~\ref{lem:H-} 
(using the stronger hardness result of Dyer and Greenhill, Theorem~\ref{thm:DG}).

The remainder of the section generalises the connected case to the case in which $H$ is not connected.

\begin{lem}\label{thm:simpleCompDichotomy}
If $H$ is a reflexive clique of size at least~$3$ then $\cComp{H}$ is $\numP$-hard.
If $H$ is an irreflexive biclique that is not a star then $\cComp{H}$ is $\numP$-hard.
\end{lem}
\begin{proof} 
Suppose that $H$ is a reflexive clique of size at least~$3$ or an irreflexive biclique that is not a star.
Recall the definitions of~$\calS_H$, $\lambda_H$ 
and weighted graph sets
(Definitions~\ref{defn:wGS}, \ref{defn:calS_H} and~\ref{defn:lambda_H}). 
Note that $(\calS_H,\lambda_H)$ is a weighted graph set.
Let $H^-$ be a graph 
obtained from~$H$ by deleting a non-loop edge. Note that 
$H^-$ is connected and it is not a reflexive clique or an irreflexive biclique.  
Thus Theorem~\ref{thm:DG} states that \cHom{H^-} is $\numP$-complete.
We will complete the proof of the Lemma by showing $\cHom{H^-}\le \cComp{H}$. 

If $H$ is a reflexive graph then the definition of $\calS_H$ 
ensures that all graphs in $\calS_H$ are reflexive. 
If $H$ is an irreflexive bipartite graph, then the definition
ensures that all graphs in $\calS_H$ are irreflexive and bipartite. 
Since $H^-$ is connected and therefore $\lambda_H(H^-)\neq0$ by Lemma~\ref{lem:H-},
we can apply Theorem~\ref{thm:CompInterpol1} to the weighted graph 
 set $(\calS_H,\lambda_H)$ with $H^-\in \calS_H$ to obtain
$\cHom{H^-}\le \cGSHom{(\calS_H,\lambda_H)}$.
By Corollary~\ref{cor:Comp=GSHom},
$\cGSHom{(\calS_H,\lambda_H)}
 \equiv \cComp{H}$.   The lemma follows.\end{proof}
 
 We use the following two definitions in Lemmas~\ref{lem:cCompCToComp1} and~\ref{lem:cCompCToComp2} and 
 in the proof of Theorem~\ref{thm:CompDichotomy}.

\begin{defn}\label{defn:calAcalB}
Let $H$ be a graph. Suppose that every connected component that has more than $j$ vertices is an irreflexive star. Suppose further that some connected component has $j$ vertices and is not an irreflexive star. Let $\calA(H)$ be the set of reflexive   components of $H$ with $j$ vertices 
and let $\calB(H)$ be the set of irreflexive non-star components of $H$ with $j$ vertices.
\end{defn} 

\begin{defn}\label{defn:H^0}
Let $L(H)$ denote the set of loops of a graph~$H$. We define the graph $H^0=(V(H),E(H)\setminus L(H))$.
\end{defn}

\begin{lem}\label{lem:cCompCToComp1}
Let $H$ be a graph in which every component is a reflexive clique or an irreflexive biclique. If $J\in \calA(H)$ then
$\cComp{J} \le \Comp{H}$.
\end{lem}
\begin{proof}
Let $H$ be a graph in which every component is a reflexive clique or an irreflexive biclique.
Let $\calA(H)=\{A_1,\dots, A_k\}$.  It follows from the definition of $\calA(H)$ 
that all elements of $\calA(H)$ are 
reflexive cliques of some size~$j$  (the same $j$ for all graphs in $\calA(H)$).

If $j\le 2$, the statement of the lemma is trivially true, since
Lemma~\ref{lem:LCompFP} shows that
\Comp{A_i} is in $\FP$, so the restricted problem
\cComp{A_i} is also in $\FP$.

Now assume $j\ge 3$.
Suppose without loss of generality that $J=A_1$.
Let
 $G$ be a (connected) input to $\cComp{J}$.
For all $i\in [k]$, let $H\setminus A_i$ be the graph constructed from $H$ by deleting the connected component $A_i$. Using Definition~\ref{defn:H^0} we define the (irreflexive) graph $G'=(H\setminus J \oplus G)^0$ as an input to \Comp{H}.
Intuitively, to form $G'$ from $H$ we replace the connected component $J$ with the graph $G$, then we delete all loops.
We will prove the following claim.

\medskip 
\noindent {\bf Claim:  
Let $h\from V(G')\to V(H)$ be a compaction from $G'$ to $H$. Then the restriction $h\vert_{V(G)}$ is a compaction from $G$ onto an element of $\calA(H)$.}
\medskip

\noindent{\bf Proof of the claim:}\quad 
As $h$ is a homomorphism, it maps each connected component of $G'$ to a connected component of $H$. As, furthermore, $h$ is a compaction and $G'$ and $H$ have the same number of connected components, it follows that there exist connected components $C_1, \dots, C_k$ of $G'$ such that for $i\in [k]$, $h\vert_{V(C_i)}$ is a compaction from $C_i$ onto $A_i$. To prove the claim, we show that $G$ is an element of $\calC=\{C_1, \dots, C_k\}$. 
In order to use all vertices of a graph in $\calA(H)$, 
i.e.\ a  reflexive size-$j$ clique, 
a graph in $\calC$ has to have at least $j$ vertices itself. Therefore and by the construction of $G'$, an element of $\calC$ can only be one of the following:
\begin{itemize}
\item a  clique with $j$ vertices,
\item a  biclique with $j$ vertices,
\item a star with at least $j$ vertices
\item or the copy of $G$.
\end{itemize}
Since $j\ge 3$, it is easy to see that there is no compaction from a star onto a 
clique with $j$ vertices. In order to compact onto a  reflexive clique of size~$j$, 
an element of $\calC$ also has to have at least $j(j-1)/2$ edges. Thus, $\calC$ does not contain any  bicliques. Finally, there are only $k-1$ connected components in $G'$ that are 
$j$-vertex cliques other than (possibly) $G$. Therefore, $G$ has to be an element of $\calC$, which proves the claim. {\bf (End of the proof of the claim.)}
\medskip

\noindent Using the notation from Definition~\ref{defn:H^0}, the claim implies
\begin{equation}
\comp{G'}{H} = \sum_{i=1}^k \comp{G}{A_i}\cdot \comp{(H\setminus A_i)^0}{H\setminus A_i}. \label{equ:calA2}
\end{equation}
We can simplify the expression (\ref{equ:calA2}) using the fact that all elements of $\calA(H)$ are 
reflexive size-$j$ cliques: 
\[
\comp{G'}{H} =k\cdot \comp{G}{J} \cdot \comp{(H\setminus J)^0}{H\setminus J}.
\]
As $\comp{(H\setminus J)^0}{H\setminus J}$ does not depend on $G$, it can be computed in constant time. Thus, using a single \Comp{H} oracle call we can compute $\comp{G}{J}$ in polynomial time as required.
\end{proof}

\begin{lem}\label{lem:cCompCToComp2}
Let $H$ be a graph in which every component is a reflexive clique or an irreflexive biclique. If $\calA(H)$ is empty but $\calB(H)$ is non-empty, then there exists a   component $J\in \calB(H)$ such that
$\cComp{J} \le \Comp{H}$.
\end{lem}
\begin{proof}
The proof is similar to that of Lemma~\ref{lem:cCompCToComp1}. For completeness, we give the details.
By Definition~\ref{defn:calAcalB} the elements of $\calB(H)$ are of the form $K_{a,b}$ with $a+b=j$ for some fixed $j$. As stars are excluded from $\calB(H)$, we have $a,b \ge 2$. Let $\calB^{\max}(H)$ denote the set of graphs with the maximum number of edges in $\calB(H)$. The elements of $\calB^{\max}(H)$ are pairwise isomorphic since the number of edges of a $K_{a,b}$ is $a\cdot b=a(j-a)$ and this function is strictly increasing for $a\leq j/2$.
For concreteness, fix $a$ and $b$ so that each $J\in \calB^{\max}(H)$ is isomorphic to $K_{a,b}$.
Let $\calB^{\max}(H)=\{B_1,\dots, B_k\}$. Take $J=B_1$.

For all $i\in [k]$, let $H\setminus B_i$ be the graph constructed from $H$ by deleting the connected component $B_i$.  Let  
$G'=(H\setminus J \oplus G)^0$ be an input to \Comp{H}. 
We will prove the following claim.
\medskip

\noindent {\bf Claim:  
Let $h\from V(G')\to V(H)$ be a compaction from $G'$ to $H$. Then the restriction $h\vert_{V(G)}$ is a compaction from $G$ onto an element of $\calB^{\max}(H)$.}
\medskip

\noindent{\bf Proof of the claim:}\quad 
As $h$ is a homomorphism, it maps each connected component of $G'$ to a connected component of $H$. As, furthermore, $h$ is a compaction and $G'$ and $H$ have the same number of connected components, it follows that there exist connected components $C_1, \dots, C_k$ of $G'$ such that for $i\in [k]$, $h\vert_{V(C_i)}$ is a compaction from $C_i$ onto $B_i$. To prove the claim, we show that $G$ is an element of $\calC=\{C_1, \dots, C_k\}$. In order to compact onto a graph in $\calB^{\max}(H)$, a graph in $\calC$ has to have at least $j$ vertices and $a\cdot b$ edges itself. By the construction of $G'$ and the fact that $\calA(H)$ is empty, a connected component in $G'$ with at least $j$ vertices and $a\cdot b$ edges can only be one of the following:
\begin{itemize}
\item a  biclique $K_{a,b}$,
\item a star with at least $j$ vertices and at least $a\cdot b$ edges
\item or the copy of $G$.
\end{itemize}
Since $a,b\ge 2$, it is easy to see that there is no compaction from a star onto a $K_{a,b}$. Finally, there are only $k-1$ connected components in $G'$ that are  bicliques of the form $K_{a,b}$ other than (possibly) $G$. Therefore, $G$ has to be an element of $\calC$, which proves the claim. {\bf (End of the proof of the claim.)}
\medskip

\noindent Using the notation from Definition~\ref{defn:H^0}, the claim implies
\begin{equation}
\comp{G'}{H} = \sum_{i=1}^k \comp{G}{B_i}\cdot \comp{(H\setminus B_i)^0}{H\setminus B_i}. \label{equ:calB2}
\end{equation}
We can simplify the expression (\ref{equ:calB2}) using   the fact that all elements of $\calB^{\max}(H)$ are of the form $K_{a,b}$:
\[
\comp{G'}{H} =k\cdot \comp{G}{J} \cdot \comp{(H\setminus J)^0}{H\setminus J}.
\]
As $\comp{(H\setminus J)^0}{H\setminus J}$ does not depend on $G$, it can be computed in constant time. Thus, using a single \Comp{H} oracle call we can compute $\comp{G}{J}$ in polynomial time as required.
\end{proof}

Finally, we prove the main theorem of this section, which we restate at this point.
{\renewcommand{\thethm}{\ref{thm:CompDichotomy}}
\begin{thm}
\ThmComp
\end{thm}
\addtocounter{thm}{-1}
}
\begin{proof}
The membership of \Comp{H} in $\numP$ is straightforward.
We distinguish between a number of cases depending on the graph $H$.

Case 1: Suppose that every connected component of $H$ is an irreflexive star
or a reflexive clique of size at most~$2$. Then \LComp{H} is in $\FP$ by Lemma~\ref{lem:LCompFP}.

Case 2:  Suppose that $H$ contains a component that is not a reflexive clique or an irreflexive biclique.
Then the hardness of $\Hom{H}$ (from Theorem~\ref{thm:DGorig})   together with
the  reduction $\Hom{H}\le \Comp{H}$ (from Theorem~\ref{thm:HomToComp})
implies  that \Comp{H} is $\numP$-hard.
The hardness of \LComp{H} follows from the trivial reduction from \Comp{H} to \LComp{H}.

Case 3: Suppose that the components of $H$ are reflexive cliques or irreflexive bicliques and  that $H$  
contains at least one component that is not an irreflexive star
or a reflexive clique of size at most~$2$.  
Every graph $J\in \calA(H) \cup \calB(H)$   is
a reflexive clique of size at least~$3$ or an irreflexive biclique that is not a star.
By Lemma~\ref{thm:simpleCompDichotomy},  \cComp{J} is $\numP$-complete. 
Finally, as $\calA(H) \cup \calB(H)$ is non-empty, we can use either Lemma~\ref{lem:cCompCToComp1} or Lemma~\ref{lem:cCompCToComp2} to obtain the existence of $J\in \calA(H) \cup \calB(H)$ with $\cComp{J} \le \Comp{H}$.
This implies that \Comp{H} is \#P-hard.
As in Case~2, the hardness of \LComp{H} follows from the trivial reduction from \Comp{H} to \LComp{H}. \end{proof}

\section{Counting Surjective Homomorphisms}\label{sec:SHom}
 
The proof of Theorem~\ref{thm:SHom} is divided into two sections. The first of these deals with tractable
cases and the second deals with hardness results and also contains the proof of the final  theorem.
Taken together,  Theorem~\ref{thm:SHom} and Dyer and Greenhill's Theorem~\ref{thm:DGorig} show that the problem
of counting surjective homomorphisms to a fixed graph~$H$ has the same complexity characterisation as
the problem of counting all homomorphisms to~$H$.

Section~\ref{sec:UniSHom} shows that this equivalence disappears in the uniform case, where $H$
is part of the input, rather than being a fixed parameter of the problem. 
Specifically, Theorem~\ref{thm:Uniform} demonstrates a setting in which counting surjective homomorphisms
is more difficult than counting all homomorphisms (assuming $\FP\neq\numP$).

\subsection{Tractability Results}\label{sec:easySHom}

\begin{thm}\label{thm:LSHomToLHom}
Let $H$ be a graph. Then
$\LSHom{H} \le \LHom{H}$.
\end{thm}
\begin{proof}
Let $H$ be fixed and $\abs{V(H)}=q$. Let $(G,\boldS)$ be an input instance of \LSHom{H}. Let $(v_1,\dots,v_n)$ be the vertices of $G$ in an arbitrary but fixed order. With respect to this ordering and with respect to a homomorphism from $G$ to $H$, let us denote by $v_{i_1}$ the first vertex of $G$ which is assigned the first new vertex of $H$ ($v_{i_1} = v_1$), $v_{i_2}$ the first vertex of $G$ which is assigned the second new vertex of $H$ and so on. Every surjective homomorphism from $G$ to $H$ contains exactly one subsequence $\boldv=(v_{i_1},\dots,v_{i_q})$ and every homomorphism containing such a subsequence is surjective. The number of subsequences is bounded from above by $\binom{n}{q}$. Let $\sigma\from \boldv \to V(H)$ be an assignment of the vertices of $H$ to the vertices in $\boldv$. There are $q!$ such assignments. We call $\psi=(\boldv,\sigma)$ a \emph{configuration} of $G$ and $\Psi(G)$ the set of all configurations for the given $G$.
For every such configuration $\psi$ we create a \LHom{H} instance $(G,\boldS^\psi)$ with $\boldS^\psi = \{S^\psi_{v_i} \subseteq V(H)\ :\ i\in [n]\}$ and
\[
S^\psi_{v_i}=
\begin{cases}
S_{v_i} \cap \{\sigma(v_{i_j})\}, &\text{ if }i=i_j \text{ for }j\in[q]\\
S_{v_i} \cap \{\sigma(v_{i_1}),\dots, \sigma(v_{i_j})\}, &\text{ for }i_{j}<i<i_{j+1}.
\end{cases}
\]
Intuitively, we use lists to ``pin'' the vertices in $\boldv$ to the vertices assigned by $\sigma$ and to prohibit the remainder of the vertices of $G$ from being mapped to new vertices of $H$.
Then
\[\sur{(G,\boldS)}{H}  = \sum_{\psi \in \Psi(G)} \hom{(G,\boldS^\psi)}{H}\]
We can compute $\sur{(G,\boldS)}{H}$ by making a \LHom{H} oracle call for every instance $(G,\boldS^\psi)$ and adding the results. The number of oracle calls $|\Psi(G)|$ is bounded from above by the polynomial $q!\binom{n}{q}\le n^q$.
\end{proof}

\begin{cor}\label{cor:easyLSHom}
Let $H$ be a graph. If every connected component of $H$ is  a reflexive  clique or an irreflexive  biclique then $\LSHom{H}$ is in $\FP$.
\end{cor}
\begin{proof}
The statement follows directly from Theorem~\ref{thm:LSHomToLHom} using  Dyer and Greenhill's dichotomy from Theorem~\ref{thm:DGorig}.
\end{proof}

\subsection{Hardness Results}\label{sec:hardSHom}

The following result and proof are very similar to that of Theorem~\ref{thm:HomToComp} and Lemma~\ref{lem:HomToComp2}, respectively. For completeness, we repeat the proof in detail.
\begin{thm}\label{thm:HomToSHom}
Let $H$ be a graph. Then
$\Hom{H} \le \SHom{H}$.
\end{thm}
\begin{proof}
Let $\abs{V(H)}=q$ and $G$ be an input to \Hom{H}. We design a graph $G_t= G\oplus W_t$ as an input to the problem \SHom{H} by adding a set $W_t$ of $t$ new isolated vertices to the graph $G$. 

We introduce some additional notation. Let $S^k(G)$ be the number of homomorphisms $\sigma$ from $G$ to $H$ that use exactly $k$ of the vertices of $H$. Let $\{w_1, \dots, w_k\}$ be a set of $k$ arbitrary but fixed vertices from $H$. We define $N^k(W_t)$ as the number of homomorphisms $\tau$ from $W_t$ to $H$ such that $\{w_1, \dots, w_k\}$ are amongst the vertices used by $\tau$. The particular choice of vertices $\{w_1, \dots, w_k\}$ is not important when counting homomorphisms from a set of isolated vertices---$N^k(W_t)$ only depends on the numbers $k$ and $t$.

We observe that, for each surjective homomorphism $\gamma\from V(G_t) \to V(H)$, the restriction $\gamma\vert_{V(G)}$ uses a subset $V'\subseteq V(H)$ of vertices and does not use any vertices outside of $V'$. Suppose that $V'$ has cardinality $\abs{V'}=q-k$ for some $k\in \{0,\dots, q\}$. Then $\gamma\vert_{W_t}$ uses at least the remaining $k$ fixed vertices of $H$.

Therefore, we obtain the following linear equation for a fixed $t\ge 0$:
\[
\underbrace{\sur{G_t}{H}}_{b_t}
=\sum_{k=0}^q \underbrace{S^{q-k}(G)}_{x_k}\underbrace{N^k(W_t)}_{a_{t,k}}.
\]

By choosing $q+1$ different values for the parameter $t$ we obtain a system of linear equations. Here, we choose $t=0,\dots ,q$. Then the system is of the form $\boldb= \boldA\boldx$ for 
\begin{equation*}
\boldb=
\begin{pmatrix}
b_0\\
\vdots\\
b_q
\end{pmatrix}
\qquad
\boldA=
\begin{pmatrix}
a_{0,0} &\dots & a_{0,q}\\
\vdots & \ddots & \vdots\\
a_{q,0} &\dots & a_{q,q}\\
\end{pmatrix}
\qquad
\text{and}
\qquad
\boldx= \begin{pmatrix}
x_0\\
\vdots\\
x_q
\end{pmatrix}.
\end{equation*}
Note, that the vector $\boldb$ can be computed using $q+1$ \SHom{H} oracle calls. Further, 
\[\sum_{k=0}^q x_k = \sum_{k=0}^q S^{q-k}(G)= \sum_{k=0}^q S^k(G) = \hom{G}{H}.\]
Thus, determining $\boldx$ is sufficient for computing the sought-for $\hom{G}{H}$. It remains to show that the matrix $\boldA$ is of full rank and is therefore invertible.

For $t<k$, clearly $a_{t,k} = N^k(W_t) = 0$. Further, for the diagonal elements we have $a_{t,t}= N^t(W_t)=t!$ for $t\in \{0,\dots, q\}$. Hence, 
\begin{align*}
\boldA =
\begin{pmatrix}
1 & 0 &\cdots &0\\
\ast & 1! & \ddots &\vdots\\
\vdots & \ddots & \ddots & 0\\
\ast & \cdots &\ast & q!
\end{pmatrix}
\end{align*}
is a triangular matrix with non-zero diagonal entries, which completes the proof.
\end{proof}

{\renewcommand{\thethm}{\ref{thm:SHom}}
\begin{thm}
\ThmSHom
\end{thm}
\addtocounter{thm}{-1}
}
\begin{proof}
The easiness result  follows from Corollary~\ref{cor:easyLSHom} using the trivial
reduction $\SHom{H} \leq \LSHom{H}$.
The hardness result follows from the same trivial reduction, along with Theorem~\ref{thm:HomToSHom} and the dichotomy for \Hom{H} from~Theorem~\ref{thm:DGorig}. 

\end{proof}

\subsection{The Uniform Case}\label{sec:UniSHom}

We have seen from Theorems~\ref{thm:DGorig} and~\ref{thm:SHom}
that the problem of counting homomorphisms to
a fixed graph~$H$ has the same complexity as the problem of counting \emph{surjective} homomorphisms to~$H$.

Nevertheless, there are scenarios in which counting problems involving surjective homomorphisms
are more difficult than those involving unrestricted homomorphisms.
To illustrate this point, we consider the following 
\emph{uniform} homomorphism-counting problems.
Motivated by terminology from constraint satisfaction, we use ``uniform'' to indicate that the target graph~$H$
is part of the input, rather than being a fixed parameter.\bigskip

 \noindent\begin{tabular}{ll} 
\textbf{Name:} $\UniHom$. & \textbf{Name:} $\UniSHom$.  
\tabularnewline[-0.05cm]
\textbf{Input:}  Irreflexive graph $G$ & \textbf{Input:} Irreflexive graph $G$  
\tabularnewline[-0.05cm]
whose components are cliques & whose components are cliques 
\tabularnewline[-0.05cm]
and reflexive graph $H$ & and reflexive graph $H$
\tabularnewline[-0.05cm]
whose components are cliques. & whose components are cliques. 
\tabularnewline[-0.05cm]
\textbf{Output:} $ \hom{G}{H}$. & \textbf{Output:}  $\sur{G}{H}$.    
\tabularnewline[-0.05cm]
\end{tabular}\bigskip

The main result of this section is the following theorem.
\newcommand{\ThmUniform}{
$\UniHom$ is in $\FP$ but $\UniSHom$ is $\numP$-complete.
}\begin{thm}\label{thm:Uniform}
\ThmUniform
\end{thm}

In order to prove Theorem~\ref{thm:Uniform}, we define 
a counting variant of the subset sum problem.
Given a set of integers $\calA=\{a_1,\dots,a_n\}$ and an integer $b$ let $S{(\calA,b)}$, 
be the number of subsets $\calA'\subseteq \calA$ such that the sum of the elements in $\calA'$ is equal to $b$. 
The counting problem is stated as follows.
\prob{\SubSum.}{A set of positive integers $\calA=\{a_1,\dots,a_n\}$ and a positive integer $b$.}{$S{(\calA,b)}$.} 
It is well known that $\SubSum$ is $\numP$-complete (see for instance the textbook by Papadimitriou~\cite[Theorems 9.9, 9.10 and 18.1]{Papadimitriou1994}). Thus, Theorem~\ref{thm:Uniform} follows immediately from
Lemmas~\ref{lem:UniHom} and~\ref{lem:SubSum}.

\begin{lem}\label{lem:UniHom}
$\UniHom$ is in $\FP$.
\end{lem}
\begin{proof}
Let $G$ and $H$ be an input instance of $\UniHom$.
Let $k$ be the number of connected components of $G$ and let $a_1,\dots ,a_k$ be the number of vertices of these components, respectively. Let $H$ have $q$ connected components with $b_1, \dots, b_q$ vertices, respectively. Then, as all components are 
cliques  and $H$ is reflexive, 
\[\hom{G}{H} = \prod_{i=1}^k \sum_{j=1}^q b_j^{a_i}.\]
Thus, it is easy to compute $\hom{G}{H}$.
\end{proof}

\begin{lem}\label{lem:SubSum}
$\SubSum \le \UniSHom$.
\end{lem}
\begin{proof}
Let $\calA=\{a_1, \dots, a_k\}$, $b$ be an input instance of \SubSum. We define $N=\sum_{i=1}^k a_i$. Now, we design a polynomial time algorithm to determine $S(\calA,b)$ using an oracle for \UniSHom. If $N<b$, we have $S(\calA,b)=0$. Now assume $N\ge b$. We create an input of $\UniSHom$ as follows. We set $G$ to be an irreflexive graph with a connected component $G_i$ for each $i\in[k]$, where $G_i$ is a  clique with $a_i$ vertices. Furthermore, we set $H$ to be a reflexive graph with two connected components $H_1$ and $H_2$. Let $H_1$ be a  clique with $b$ vertices and let $H_2$ be a  clique with $N-b$ vertices. By $\stirling{n}{k}$ we denote the Stirling number of the second kind, i.e.\ the number of partitions of a set of $n$ elements into $k$ non-empty subsets. By definition, we have $\stirling{n}{k}=0$ if $n<k$.

Let $h\from V(G) \to V(H)$ be a homomorphism from $G$ to $H$ and let $b'$ be the number of vertices of $G$ that are mapped to the connected component $H_1$. Note that $h$ has to map each connected component of $G$ to a connected component of $H$. By the construction of $G$, this implies that there exists a subset $\calA'\subseteq \calA$ such that the sum of elements in $\calA'$ is equal to $b'$.
Furthermore, as all connected components of $G$ and $H$ are  cliques and $H$ is reflexive, the number of surjective homomorphisms from $G$ to $H$ that assign exactly $b'$ fixed vertices to $H_1$ is equal to the number of surjective mappings from $[b']$ to $[b]$, which is $b!\stirling{b'}{b}$. 
Therefore, we can express $\sur{G}{H}$ as follows.
\begin{align}
\sur{G}{H}=\sum_{b'=0}^N S(\calA,b')\cdot b!\stirling{b'}{b}\cdot (N-b)!\stirling{N-b'}{N-b}, \label{equ:SubSum1}
\end{align}
where the factor $(N-b)!\stirling{N-b'}{N-b}$ corresponds to the number surjective mappings from the remaining $N-b'$ fixed vertices of $G$ to the component $H_2$.
Finally, we use the fact that the summands in (\ref{equ:SubSum1}) are non-zero only if $b'\ge b$ and $N-b'\ge N-b$, which implies $b'=b$. Thus,
\begin{align*}
\sur{G}{H}&= S(\calA,b)\cdot b!\stirling{b}{b}\cdot (N-b)!\stirling{N-b}{N-b}\\
&=b!(N-b)!\cdot S(\calA,b).
\end{align*}
\end{proof}

\section{Addendum: A Dichotomy for Approximately Counting Homomorphisms with Surjectivity Constraints}

The following standard definitions are taken from~\cite[Definitions 11.1, 11.2, Exercise 11.3]{Mitzenmacher2017}.
A randomised algorithm gives an $(\epsilon,\delta)$-approximation
for the value~$V$ if the output~$X$ of the algorithm satisfies $\Pr(|X-V| \leq \epsilon V) \geq 1-\delta$.
A \emph{fully polynomial randomised approximation scheme} (FPRAS)
for a problem~$V$ is a randomised algorithm which, given an input~$x$ and a parameter 
$\epsilon \in (0,1)$, outputs an $(\epsilon,1/4)$-approximation to~$V(x)$ in time that is polynomial in~$1/\epsilon$
and the size of the input~$x$.
The concept of an 
approximation-preserving reduction (AP-reduction) between counting
problems  was introduced by Dyer, Goldberg, Greenhill and Jerrum~\cite{DGGJApprox}. 
We will not need the detailed
definition here, but the definition has the property that if there is an AP-reduction from
problem~$A$ to problem~$B$ (written as $A\leap B$) then
this reduction, together with an FPRAS for~$B$, yields an FPRAS for~$A$. 
The problem $\bis$, which is the problem of counting the independent sets of a bipartite graph, comes up frequently
in approximate counting because  it is complete with respect to AP-reductions in an intermediate 
complexity class. It is not believed to have an FPRAS.
Galanis,  Goldberg and Jerrum~\cite{Galanis2016} gave a dichotomy for
the problem of \emph{approximately}
 counting homomorphisms in the connected case, in terms of $\bis$.

\begin{thm}[\cite{Galanis2016}]\label{thm:HomBIS}
Let $H$ be a connected graph. If $H$ is a reflexive clique or an irreflexive biclique, then there is an FPRAS for \Hom{H}. Otherwise, $\bis\leap\Hom{H}$.
\end{thm}

In this addendum we give a similar dichotomy for approximately counting homomorphisms with surjectivity constraints\footnote{When $H$ is not connected, the complexity of approximate counting is open even for counting homomorphisms. Hence we do not address this case here.}.
The tractability part of the following theorem follows from 
Theorem~\ref{thm:SHom}, Corollary~\ref{cor:retdichotomy} and from Lemma~\ref{lem:CompFPRAS} 
below. The
$\bis$-hardness   follows from Theorem~\ref{thm:HomBIS} and 
 from the reductions in Lemmas~\ref{lem:APHomToSHom},~\ref{lem:APSHomToComp} and~\ref{lem:APHomToRet}.

\begin{thm}\label{thm:ApproxDicho}
Let $H$ be a connected graph. If $H$ is a reflexive clique or an irreflexive biclique, then there is an FPRAS for  \SHom{H}, \Ret{H} and \Comp{H}. Otherwise, 
each of these problems is $\bis$-hard under approximation-preserving reductions.
\end{thm} 

\begin{lem}\label{lem:CompFPRAS}
Let $H$ be a reflexive clique or an irreflexive biclique. Then there is an FPRAS for \Comp{H}.
\end{lem}

\begin{proof}

Let $H$ be a reflexive clique or an irreflexive biclique with $q$ vertices and $p$ edges.
Our goal is give an FPRAS for \Comp{H}.
 
First, we show that we can assume without loss of generality that every input~$G$
to \Comp{H} has no isolated vertices.
To see this, suppose instead that  $G$ is of the form $G'\oplus I$ where $I$ is the set of isolated vertices in $G$. 
As $H$ is connected, we have $\comp{G}{H}=q^{\abs{I}}\comp{G'}{H}$. 
Thus, an estimate of the number of compactions from $G'$ to $H$ will
immediately enable us to approximately count compactions from $G$ to $H$.

From now on we restrict attention to inputs $G$ which have no isolated vertices.
We use $\calH(G,H)$ to denote the set of homomorphisms from $G$ to $H$.  
   
\noindent{\bf Case 1. $H$ is a reflexive clique.}

Let $G$ be a size-$n$ input to \Comp{H}.
Then $\hom{G}{H}=q^n$. 
If there is a  compaction from $G$ to $H$ 
then there is a set $U\subseteq V(G)$ 
with $|U| \leq 2p$
and a compaction~$\sigma$ from~$G[U]$ to~$H$. 
Each assignment of  the (at most $n-2p$) vertices in $V(G)\setminus U$
extends $\sigma$
to a  compaction from~$G$ to~$H$. 
Thus, we have $\comp{G}{H}\ge q^{n-2p}=\hom{G}{H}/q^{2p}$. Using this lower bound, it is straightforward to apply the naive Monte Carlo method (cf. \cite[Theorem 11.1]{Mitzenmacher2017}).
Hence Algorithm~\ref{algo:MC} with $c=q^{2p}$ and $\calH=\calH(G,H)$ 
gives an $(\epsilon,\delta)$-approximation for the number of compactions in $\cal{H}$.

\begin{algorithm}
\caption{If the number of compactions in $\mathcal{H}$ is at least $| \calH|/c$ then  
by 
 \cite[Theorem 11.1]{Mitzenmacher2017}
this algorithm gives an $(\epsilon,\delta)$-approximation
for the number of compactions in $\mathcal{H}$.}
\label{algo:MC}
\begin{algorithmic}
\Require Irreflexive graph $G$, $\epsilon\in (0,1)$ and $\delta\in (0,1)$.
\State   $ m = \left\lceil c  3 \ln(2/\delta)/\epsilon^2\right\rceil$.
\State Choose $m$ samples independently and uniformly at random from $\calH$.
\State Let $X_1,\dots, X_m$ be the corresponding indicator random variables, where $X_i$ takes value $1$ \State if the $i$th sample is a compaction and $0$ otherwise.
\State $\displaystyle Y=\frac{\abs{\calH}}{m} \sum_{i=1}^m X_i$.
\Ensure $Y$
\end{algorithmic}
\end{algorithm}

If there are no compactions in $\mathcal{H}$ then the algorithm answers~$0$.
Otherwise, the number of compactions in $\mathcal{H}$ is at least $|\calH|/c$,
so the algorithm gives an $(\epsilon,\delta)$-approximation.

When the algorithm is run with $\delta=1/4$, the running time is at most a polynomial in $n$ and
$1/\epsilon$ because $m$ is at most a polynomial in $1/\epsilon$ and
the basic tasks (choosing a sample from $\mathcal{H}$, determining whether a sample is a compaction,
and computing $|\mathcal{H}|=q^n$) can all be done in $\text{poly}(n)$ time. 
Thus, the algorithm gives an FPRAS for \Comp{H}.

\noindent{\bf Case 2. $H$ is an irreflexive biclique.}

Let the bipartition of $V(H)$ be   $(L_H,R_H)$  
where $\ell_H = |L_H| \leq |R_H| = r_H$. 
We can assume that $\ell_H \geq 1$, otherwise counting compactions to~$H$ is trivial.

Without loss generality, we can assume that inputs $G$ to \Comp{H}
are bipartite (as well as having no isolated vertices).
(If $G$ is not bipartite, then $\comp{G}{H}=0$.)

Suppose that $G$ is an input to \Comp{H}.
Let $C_1,\dots,C_{\kappa}$ be the connected components of $G$. For each $i\in [\kappa]$, let $(L_i,R_i)$ be a fixed bipartition of $C_i$ such that 
$1\leq \ell_i = \abs{L_i}\le \abs{R_i} = r_i$.
Then $\hom{G}{H}=\prod_{i=1}^{\kappa} \left(\ell_H^{ {\ell_i}}{r_H}^{ {r_i}} + {\ell_H}^{ {r_i}}{r_H}^{ {\ell_i}}\right) \leq 
2 \prod_{i=1}^{\kappa} {\ell_H}^{ {\ell_i}}{r_H}^{{r_i}}$. 

Let $\Omega$ be the set of functions $\omega   \from [\kappa] \to \{L_H,R_H\}$.
Given $\omega\in \Omega$, we say that a homomorphism from $G$ to $H$ obeys $\omega$ if, for each $i\in [\kappa]$, the vertices of $L_i$ are assigned to vertices in $\omega(i)$.  

\noindent{\bf Case 2a. $\kappa \geq p$.}

Let $\omega$ be the function in $\Omega$ that maps every $i\in [\kappa]$ to $L_H$.
Since $G$ has no isolated vertices, each of $C_1,\ldots,C_\kappa$ has at least $2$~vertices,
so there is a compaction from~$G$ to~$H$ which obeys~$\omega$.

As in Case~1, there is a set $U\subseteq V(G)$ of size at most $2p$
such that 
there is a compaction  $\sigma$ from $G[U]$ to~$H$
that obeys the restriction of~$\sigma$ to~$U$.
Every assignment of the vertices in $V(G)\setminus U$  that obeys~$\omega$
yields an $\omega$-obeying compaction from~$G$ to~$H$.
Since $r_H\geq \ell_H$, we obtain the lower bound
 
\[
\comp{G}{H}\ge  
\frac{1}{{(r_H)}^{2p}}\prod_{i=1}^{\kappa} {\ell_H}^{{\ell_i}}{r_H}^{{r_i}}\ge \frac{\hom{G}{H}}{2{(r_H)}^{2p}}.
\]
By the same arguments as 
in Case~1, Algorithm~\ref{algo:MC} with $c=2{(r_H)}^{2p}$ and $\calH=\calH(G,H)$ gives an $(\epsilon,\delta)$-approximation for 
the number of compactions in $\calH$.
When the algorithm is run with $\delta=1/4$, the running time is at most a polynomial in $|V(G)|$ and $1/\epsilon$,
so it can be used in an FPRAS for inputs $G$ with $\kappa \geq p$.

\noindent{\bf Case 2b. $\kappa < p$.}

For each $\omega\in \Omega$, let $\calH_\omega(G,H)$ be the set of homomorphisms obeying $\omega$,  and let $N_\omega(G\to H)$ and $N^{\text{comp}}_\omega(G\to H)$ be the number of homomorphisms and compactions obeying $\omega$, respectively. 
Given a compaction that obeys $\omega$ we obtain a lower bound as before:
\[
N^{\text{comp}}_\omega(G\to H)\ge \frac{1}{{(r_H)}^{2p}}\prod_{i=1}^{\kappa} \abs{\omega(i)}^{{\ell_i}}(\abs{V(H)}-\abs{\omega(i)})^{{r_i}}= \frac{N_\omega(G\to H)}{{(r_H)}^{2p}}.
\]
Now Algorithm~\ref{algo:MC} with 
$c={(r_H)}^{2p}$ and $\calH=\calH_\omega(G,H)$
gives an $(\epsilon,\delta)$-approximation for the number of compactions in 
$\calH_\omega(G,H)$.
Taking $\delta = 1/(4 \cdot 2^{\kappa})$ and summing over the $2^{\kappa}< 2^p$
functions $\omega\in \Omega$,
we obtain an $(\epsilon,1/4)$-approximation for the number of compactions in $\calH(G,H)$.
 The running time of each call to Algorithm~\ref{algo:MC} 
 is at most a polynomial in $|V(G)|$ and $1/\epsilon$.
 Thus, putting the cases together, we get an FPRAS for \Comp{H}.\end{proof}

\begin{lem}\label{lem:APHomToSHom}
Let $H$ be a graph. Then $\Hom{H} \leap \SHom{H}$.
\end{lem}
\begin{proof}
Let $q=|V(H)|$. Given any positive integer~$t$, let $s_{t,q}$ denote the number
of surjective functions from $[t]$ to $[q]$.
Clearly, $s_{t,q} \geq q^t - 2^q {(q-1)}^t$,
since the range of
every non-surjective function from~$[t]$ to~$[q]$ is 
a proper subset of $[q]$, and there are most $2^q$ of these.
Also,  the number of functions from $[t]$ onto  this subset is at most ${(q-1)}^t$.

Given any $n$-vertex input $G$ to the problem $\Hom{H}$,
let 
$$t = \lceil
\log(5 q^n 2^q)/\log(q/(q-1)
\rceil.$$
Clearly, $t=O(n)$, and $t$ can be computed in time $\mathrm{poly}(n)$.
Note that
\begin{equation}\label{eq:easy}
{\left(\frac{q}{q-1}\right)}^t \geq 5 q^n 2^q \geq 4 q^n2^q + 2^q.
\end{equation}
Let $G_t$ be the graph constructed from~$G$ by adding 
a set $I_t$ of $t$ isolated vertices that are distinct from the vertices in~$V(G)$.
We claim that
$$s_{t,q} \hom{G}{H} \leq \sur{G_t}{H} \leq s_{t,q} \hom{G}{H} + (q^t - s_{t,q}) q^n.$$
To see this, note that any homomorphism from $G$ to $H$, together with a
surjective homomorphism from the $I_t$ to $V(H)$,
constitutes a surjective homomorphism from $G_t$ to $H$.
Any other surjective homomorphism from $G_t$ to $H$
consists of a non-surjective homomorphism
from $I_t$ to $H$ (and there are $q^t - s_{t,q}$ of these) together
with some homomorphism from $G$ to $H$ (and there are at most $q^n$ of these).
Dividing through by $s_{t,q}$ and applying 
our lower bound for $s_{t,q}$ and then
inequality~\eqref{eq:easy}, we have
\begin{align}\nonumber
\hom{G}{H} \leq \frac{ \sur{G_t}{H}}{s_{t,q}} & \leq \hom{G}{H} + \left(\frac{q^t - s_{t,q}}{s_{t,q}}\right) q^n\\ \nonumber
& \leq \hom{G}{H} + \frac{ 2^q {(q-1)}^t  q^n  }{  q^t - 2^q {(q-1)}^t } \\ \nonumber
& = \hom{G}{H} + \frac{    q^n  }{  \frac{q^t}{2^q{(q-1)}^t} -  1} \\ 
& \leq \hom{G}{H} + \frac{  1  }{  4}. \label{eq:red}
\end{align}
Given Equation~\eqref{eq:red}, the proof of
\cite[Theorem 3]{DGGJApprox} shows that
to approximate $\hom{G}{H}$ with accuracy $\varepsilon$, we need only
use the oracle to obtain an approximation $\widehat{S}$ for 
$\sur{G_t}{H}$ with accuracy $\epsilon/21$. We can then
return the floor of $\widehat{S}/s_{t,q}$.
The only remaining issue is how to compute $s_{t,q}$.
However, it is easy to do this in 
time $\mathrm{poly}(t) = \mathrm{poly}(n)$ since  $s_{t,q} = \stirling{t}{q} q! = \sum_{j=0}^{q}{(-1)}^{q-j} \binom{q}{j} j^t$,
where $\stirling{t}{q}$ is a Stirling number of the second kind.\end{proof}

\begin{lem}\label{lem:APSHomToComp}
Let $H$ be a connected graph. Then $\Hom{H} \leap \Comp{H}$.
\end{lem}
\begin{proof}
If not explicitly defined otherwise, we use the same notation and observations as in the proof of Lemma~\ref{lem:APHomToSHom}. In addition let $p$ be the number of non-loop edges in $H$ and $c_{t,p}=2^t s_{t,p}$. 
If $G$ is an input to \Hom{H} of size $n$, $G_t$ is the graph constructed from $G$ by adding a set of $t$ isolated edges distinct from the edges in $G$. If $H$ is a graph of size $1$ the statement of the lemma clearly holds. If otherwise $H$ is a connected graph of size at least $2$, every homomorphism that uses all non-loop edges of $H$ is also surjective and therefore a compaction. Thus, we obtain
\[c_{t,p} \hom{G}{H} \leq \comp{G_t}{H} \leq c_{t,p} \hom{G}{H} + (2^tp^t - c_{t,p}) q^n.\]
Dividing through by $c_{t,p}$ gives
\[\hom{G}{H} \leq \frac{\comp{G_t}{H}}{c_{t,p}} \leq  \hom{G}{H} + \left(\frac{p^t - s_{t,p}}{s_{t,p}}\right) q^n.\]
If we choose $t = \lceil
\log(5 q^n 2^p)/\log(p/(p-1)
\rceil$  the remainder of this proof is analogous to that of Lemma~\ref{lem:APHomToSHom}.
\end{proof}

\begin{lem}\label{lem:APHomToRet}
Let $H$ be a graph. Then $\Hom{H} \leap \Ret{H}$.
\end{lem}
\begin{proof}
Let $q=\abs{V(H)}$ and $G$ be an input to \Hom{H}. Further, let $H'$ be a copy of $H$ and $(u_1,\dots,u_q)$ be the vertices of $H'$ ordered in such a way that they induce a copy of $H$. Then $\hom{G}{H} = \ret{(G\oplus H'; u_1,\ldots,u_q)}{H}$.
\end{proof}

 \bibliography{\jobname}

\begin{thebibliography}{10}

\bibitem{BodirskySurvey}
Manuel Bodirsky, Jan K\'{a}ra, and Barnaby Martin.
\newblock The complexity of surjective homomorphism problems{ -- }a survey.
\newblock {\em Discrete Applied Mathematics}, 160(12):1680 -- 1690, 2012.

\bibitem{BorgsChayesLovaszCounting}
Christian Borgs, Jennifer Chayes, L\'aszl\'o Lov\'asz, Vera~T. S\'os, and
  Katalin Vesztergombi.
\newblock Counting graph homomorphisms.
\newblock In {\em Topics in discrete mathematics}, volume~26 of {\em Algorithms
  Combin.}, pages 315--371. Springer, Berlin, 2006.

\bibitem{Borgs}
Christian Borgs, Jennifer~T. Chayes, Jeff Kahn, and L{\'{a}}szl{\'{o}}
  Lov{\'{a}}sz.
\newblock Left and right convergence of graphs with bounded degree.
\newblock {\em Random Struct. Algorithms}, 42(1):1--28, 2013.

\bibitem{BW}
Graham~R. Brightwell and Peter Winkler.
\newblock Graph homomorphisms and phase transitions.
\newblock {\em J. Comb. Theory, Ser. {B}}, 77(2):221--262, 1999.

\bibitem{HubieChen}
Hubie {Chen}.
\newblock {Homomorphisms are indeed a good basis for counting: Three
  fixed-template dichotomy theorems, for the price of one}.
\newblock {\em CoRR}, abs/1710.00234, 2017.

\bibitem{hombasis}
Radu Curticapean, Holger Dell, and D{\'{a}}niel Marx.
\newblock Homomorphisms are a good basis for counting small subgraphs.
\newblock In {\em Proceedings of the 49th Annual ACM Symposium on Theory of
  Computing}, pages 210--213, 2017.

\bibitem{HolgerDell}
Holger {Dell}.
\newblock {Note on ``The Complexity of Counting Surjective Homomorphisms and
  Compactions''}.
\newblock {\em CoRR}, abs/1710.01712, 2017.

\bibitem{Diaz2004}
Josep D\'iaz, Maria Serna, and Dimitrios Thilikos.
\newblock Recent results on parameterized ${H}$-coloring.
\newblock In J.~Ne\v{s}et\v{r}il and P.~Winkler, editors, {\em Graphs,
  Morphisms and Statistical Physics}, volume~63 of {\em DIMACS Series in
  Discrete Mathematics and Theoretical Computer Science}, 2004.

\bibitem{DGGJApprox}
Martin Dyer, Leslie~Ann Goldberg, Catherine Greenhill, and Mark Jerrum.
\newblock The relative complexity of approximate counting problems.
\newblock {\em Algorithmica}, 38(3):471--500, 2004.

\bibitem{DG}
Martin Dyer and Catherine Greenhill.
\newblock The complexity of counting graph homomorphisms.
\newblock {\em Random Structures \& Algorithms}, 17(3-4):260--289, 2000.

\bibitem{FHRetractions}
Tom{\'{a}}s Feder and Pavol Hell.
\newblock List homomorphisms to reflexive graphs.
\newblock {\em J. Combin. Theory Ser. B}, 72(2):236--250, 1998.

\bibitem{FHH}
Tom{\'{a}}s Feder, Pavol Hell, and Jing Huang.
\newblock List homomorphisms and circular arc graphs.
\newblock {\em Combinatorica}, 19(4):487--505, 1999.

\bibitem{FHJKN}
Tom\'as Feder, Pavol Hell, Peter Jonsson, Andrei Krokhin, and Gustav Nordh.
\newblock Retractions to pseudoforests.
\newblock {\em SIAM J. Discrete Math.}, 24(1):101--112, 2010.

\bibitem{FGZ2018}
Jacob Focke, Leslie~Ann Goldberg, and Stanislav \v{Z}ivn\'{y}.
\newblock {The Complexity of Counting Surjective Homomorphisms and
  Compactions}.
\newblock In {\em Proceedings of the Twenty-Ninth Annual ACM-SIAM Symposium on
  Discrete Algorithms}, SODA '18, pages 1772--1781, Philadelphia, PA, USA,
  2018. Society for Industrial and Applied Mathematics.

\bibitem{Galanis2016}
Andreas Galanis, Leslie~Ann Goldberg, and Mark Jerrum.
\newblock Approximately counting ${H}$-colorings is $\#\mathrm{BIS}$-hard.
\newblock {\em SIAM Journal on Computing}, 45(3):680--711, 2016.

\bibitem{Gobel2016}
Andreas G\"{o}bel, Leslie~Ann Goldberg, and David Richerby.
\newblock Counting homomorphisms to square-free graphs, modulo 2.
\newblock {\em ACM Trans. Comput. Theory}, 8(3):12:1--12:29, May 2016.

\bibitem{GolovachNewHardness}
Petr~A. Golovach, Matthew Johnson, Barnaby Martin, Dani{\"{e}}l Paulusma, and
  Anthony Stewart.
\newblock Surjective ${H}$-colouring: New hardness results.
\newblock {\em CoRR}, abs/1701.02188, 2017.

\bibitem{GolovachFinding}
Petr~A. Golovach, Bernard Lidick{\'{y}}, Barnaby Martin, and Dani{\"{e}}l
  Paulusma.
\newblock Finding vertex-surjective graph homomorphisms.
\newblock {\em Acta Inf.}, 49(6):381--394, 2012.

\bibitem{GolovachTrees}
Petr~A. Golovach, Dani{\"{e}}l Paulusma, and Jian Song.
\newblock Computing vertex-surjective homomorphisms to partially reflexive
  trees.
\newblock {\em Theoretical Computer Science}, 457:86 -- 100, 2012.

\bibitem{HellMiller}
P.~Hell and D.~J. Miller.
\newblock Graphs with forbidden homomorphic images.
\newblock {\em Annals of the New York Academy of Sciences}, 319(1):270--280,
  1979.

\bibitem{HNOld}
P.~Hell and J.~Ne\v{s}et\v{r}il.
\newblock Homomorphisms of graphs and of their orientations.
\newblock {\em Monatsh. Math.}, 85(1):39--48, 1978.

\bibitem{HellNesetrilOriginal}
P.~Hell and J.~Ne\v{s}et\v{r}il.
\newblock On the complexity of ${H}$-coloring.
\newblock {\em Journal of Combinatorial Theory, Series B}, 48(1):92 -- 110,
  1990.

\bibitem{Hell2004b}
P.~Hell and J.~Ne\v{s}et\v{r}il.
\newblock Counting list homomorphisms for graphs with bounded degrees.
\newblock In J.~Ne\v{s}et\v{r}il and P.~Winkler, editors, {\em Graphs,
  Morphisms and Statistical Physics}, volume~63 of {\em DIMACS Series in
  Discrete Mathematics and Theoretical Computer Science}, pages 105--112, 2004.

\bibitem{HellNesetrilBook}
P.~Hell and J.~Ne\v{s}et\v{r}il.
\newblock {\em Graphs and Homomorphisms}.
\newblock Oxford Lecture Series in Mathematics and Its Applications. OUP
  Oxford, 2004.

\bibitem{Lovasz1967}
L.~Lov\'asz.
\newblock Operations with structures.
\newblock {\em Acta Mathematica Academiae Scientiarum Hungaricae}, 1967.

\bibitem{LovaszBook}
L.~Lov{\'a}sz.
\newblock {\em Large Networks and Graph Limits}.
\newblock American Mathematical Society colloquium publications. American
  Mathematical Society, 2012.

\bibitem{MartinPaulusmaSHomC4}
Barnaby Martin and Dani{\"{e}}l Paulusma.
\newblock The computational complexity of disconnected cut and 2k2-partition.
\newblock {\em Journal of Combinatorial Theory, Series B}, 111:17 -- 37, 2015.

\bibitem{Mitzenmacher2017}
Michael Mitzenmacher and Eli Upfal.
\newblock {\em Probability and Computing: Randomization and Probabilistic
  Techniques in Algorithms and Data Analysis}.
\newblock Probability and Computing: Randomization and Probabilistic Techniques
  in Algorithms and Data Analysis. Cambridge University Press, 2017.

\bibitem{Papadimitriou1994}
Christos~M. Papadimitriou.
\newblock {\em {Computational complexity}}.
\newblock Addison-Wesley, Reading, Massachusetts, 1994.

\bibitem{VikasSICOMPCompl}
Narayan Vikas.
\newblock Computational complexity of compaction to reflexive cycles.
\newblock {\em SIAM J. Comput.}, 32(1):253--280, January 2003.

\bibitem{Vikas2004}
Narayan Vikas.
\newblock Compaction, retraction, and constraint satisfaction.
\newblock {\em SIAM Journal on Computing}, 33(4):761--782, 2004.

\bibitem{Vikas2004b}
Narayan Vikas.
\newblock Computational complexity of compaction to irreflexive cycles.
\newblock {\em Journal of Computer and System Sciences}, 68(3):473 -- 496,
  2004.

\bibitem{Vikas2005}
Narayan Vikas.
\newblock A complete and equal computational complexity classification of
  compaction and retraction to all graphs with at most four vertices and some
  general results.
\newblock {\em Journal of Computer and System Sciences}, 71(4):406 -- 439,
  2005.

\bibitem{Vikas2013}
Narayan Vikas.
\newblock Algorithms for partition of some class of graphs under compaction and
  vertex-compaction.
\newblock {\em Algorithmica}, 67(2):180--206, 2013.

\end{thebibliography}
 
 \clearpage

\section*{Appendix: Decomposition of $\comp{G}{K_{2,3}}$} 
  
In this  appendix, we work through a long example to illustrate
some of
the definitions and ideas from Section~\ref{sec:hardComp}.
We do this by 
verifying the statement of Theorem~\ref{thm:CompDecomp}
for the special case where $H=K_{2,3}$.

Of course, the theorem is already proved in the earlier sections of this paper, but
we work through this example in order to help the reader gain familiarity with the definitions.
For $H=K_{2,3}$ and a non-empty, irreflexive and connected graph $G$ we want to prove
\begin{equation}\label{equ:A1}
\comp{G}{H} = \sum_{J\in \calS_H} \lambda_H(J)\hom{G}{J}.
\end{equation}
First, we set $\calS_H=\{H_1,\dots,H_{10}\}$, cf. Figure~\ref{fig:calSK23}, as defined in Definition~\ref{defn:calS_H}. 

\begin{figure}[h]
\centering
\def\scaleFactor{.6}
\def\arraystretch{1.25}
\smallskip
\begin{tabularx}{\textwidth}{@{}c*5{>{\centering\arraybackslash}X}@{}}
& $H\cong H_{1}$ & $H_{2}$ & $H_{3}$ &$H_{4}$& $H_{5}$\\
\hline  
& 
\begin{tikzpicture}[scale=\scaleFactor, baseline=\scaleFactor*(2cm-.2\baselineskip)]

			\coordinate (L1)  at (0, 3);
			\coordinate (L2)  at (0, 1);
			\coordinate (R1)  at (3, 4);
			\coordinate (R2)  at (3, 2);
			\coordinate (R3)  at (3, 0);
			
			\fill (L1) circle[radius=3pt];
			\fill (L2) circle[radius=3pt];
			\fill (R1) circle[radius=3pt];
			\fill (R2) circle[radius=3pt];
			\fill (R3) circle[radius=3pt];
	
			\draw (R1) -- (L1) -- (R2) -- (L2) -- (R3);
			\draw (L1) -- (R3);
			\draw (L2) -- (R1);
			\addvmargin{3mm}
\end{tikzpicture}
&
\begin{tikzpicture}[scale=\scaleFactor, baseline=\scaleFactor*(2cm-.2\baselineskip)]

			\coordinate (L1)  at (0, 3);
			\coordinate (L2)  at (0, 1);
			\coordinate (R1)  at (3, 4);
			\coordinate (R2)  at (3, 2);
			\coordinate (R3)  at (3, 0);
			
			\fill (L1) circle[radius=3pt];
			\fill (L2) circle[radius=3pt];
			\fill (R1) circle[radius=3pt];
			\fill (R2) circle[radius=3pt];
			\fill (R3) circle[radius=3pt];
	
			\draw (R1) -- (L1) -- (R2) -- (L2) -- (R3);
			\draw (L1) -- (R3);
			\addvmargin{3mm}
\end{tikzpicture}
&
\begin{tikzpicture}[scale=\scaleFactor, baseline=\scaleFactor*(2cm-.2\baselineskip)]

			\coordinate (L1)  at (0, 3);
			\coordinate (L2)  at (0, 1);
			\coordinate (R1)  at (3, 4);
			\coordinate (R2)  at (3, 2);
			\coordinate (R3)  at (3, 0);
			
			\fill (L1) circle[radius=3pt];
			\fill (L2) circle[radius=3pt];
			\fill (R1) circle[radius=3pt];
			\fill (R2) circle[radius=3pt];
			\fill (R3) circle[radius=3pt];
	
			\draw (R1) -- (L1) -- (R2) -- (L2) -- (R3);
			\addvmargin{3mm}
\end{tikzpicture}
&
\begin{tikzpicture}[scale=\scaleFactor, baseline=\scaleFactor*(2cm-.2\baselineskip)]

			\coordinate (L1)  at (0, 3);
			\coordinate (L2)  at (0, 1);
			\coordinate (R1)  at (3, 4);
			\coordinate (R2)  at (3, 2);
			\coordinate (R3)  at (3, 0);
			
			\fill (L1) circle[radius=3pt];
			\fill (L2) circle[radius=3pt];
			\fill (R2) circle[radius=3pt];
			\fill (R3) circle[radius=3pt];
	
			\draw (L1) -- (R2) -- (L2) -- (R3);
			\draw (L1) -- (R3);
			\addvmargin{3mm}
\end{tikzpicture}
&
\begin{tikzpicture}[scale=\scaleFactor, baseline=\scaleFactor*(2cm-.2\baselineskip)]

			\coordinate (L1)  at (0, 3);
			\coordinate (L2)  at (0, 1);
			\coordinate (R1)  at (3, 4);
			\coordinate (R2)  at (3, 2);
			\coordinate (R3)  at (3, 0);
			
			\fill (L1) circle[radius=3pt];
			\fill (L2) circle[radius=3pt];
			\fill (R1) circle[radius=3pt];
			\fill (R2) circle[radius=3pt];
			\fill (R3) circle[radius=3pt];
	
			\draw (R1) -- (L1) -- (R2);
			\draw (L2) -- (R3);
			\draw (L1) -- (R3);
			\addvmargin{3mm}
\end{tikzpicture}
\\
\multicolumn{6}{c}{}\\
& $H_{6}$ & $H_{7}$ & $H_{8}$ &$H_{9}$& $H_{10}$\\
\hline  
&
\begin{tikzpicture}[scale=\scaleFactor, baseline=\scaleFactor*(2cm-.2\baselineskip)]

			\coordinate (L1)  at (0, 3);
			\coordinate (L2)  at (0, 1);
			\coordinate (R1)  at (3, 4);
			\coordinate (R2)  at (3, 2);
			\coordinate (R3)  at (3, 0);
			
			\fill (L1) circle[radius=3pt];
			\fill (R1) circle[radius=3pt];
			\fill (R2) circle[radius=3pt];
			\fill (R3) circle[radius=3pt];
	
			\draw (R1) -- (L1) -- (R2);
			\draw (L1) -- (R3);
			\addvmargin{3mm}
\end{tikzpicture}
&
\begin{tikzpicture}[scale=\scaleFactor, baseline=\scaleFactor*(2cm-.2\baselineskip)]

			\coordinate (L1)  at (0, 3);
			\coordinate (L2)  at (0, 1);
			\coordinate (R1)  at (3, 4);
			\coordinate (R2)  at (3, 2);
			\coordinate (R3)  at (3, 0);
			
			\fill (L1) circle[radius=3pt];
			\fill (L2) circle[radius=3pt];
			\fill (R2) circle[radius=3pt];
			\fill (R3) circle[radius=3pt];
	
			\draw (R2) -- (L1) -- (R3) -- (L2);
			\addvmargin{3mm}
\end{tikzpicture}
&
\begin{tikzpicture}[scale=\scaleFactor, baseline=\scaleFactor*(2cm-.2\baselineskip)]

			\coordinate (L1)  at (0, 3);
			\coordinate (L2)  at (0, 1);
			\coordinate (R1)  at (3, 4);
			\coordinate (R2)  at (3, 2);
			\coordinate (R3)  at (3, 0);
			
			\fill (L1) circle[radius=3pt];
			\fill (R1) circle[radius=3pt];
			\fill (R2) circle[radius=3pt];
	
			\draw (R1) -- (L1) -- (R2);
			\addvmargin{3mm}
\end{tikzpicture}
&
\begin{tikzpicture}[scale=\scaleFactor, baseline=\scaleFactor*(2cm-.2\baselineskip)]

			\coordinate (L1)  at (0, 3);
			\coordinate (L2)  at (0, 1);
			\coordinate (R1)  at (3, 4);
			\coordinate (R2)  at (3, 2);
			\coordinate (R3)  at (3, 0);
			
			\fill (L1) circle[radius=3pt];
			\fill (R1) circle[radius=3pt];
	
			\draw (R1) -- (L1);
			\addvmargin{3mm}
\end{tikzpicture}
&
\begin{tikzpicture}[scale=\scaleFactor, baseline=\scaleFactor*(2cm-.2\baselineskip)]

			\coordinate (L1)  at (0, 3);
			\coordinate (L2)  at (0, 1);
			\coordinate (R1)  at (3, 4);
			\coordinate (R2)  at (3, 2);
			\coordinate (R3)  at (3, 0);
			
			\fill (L1) circle[radius=3pt];
			
			\addvmargin{3mm}
\end{tikzpicture}
\end{tabularx}
\caption{$\calS_H=\{H_1,\dots, H_{10}\}$}
\label{fig:calSK23}
\end{figure}

Next, we recall the definitions of $\mu_H$ and $\lambda_H$ from Definitions~\ref{defn:calS_H} and~\ref{defn:lambda_H}. For $J\in \calS_H$, $\mu_H(J)$ is the number of non-empty connected subgraphs of $H$ that are isomorphic to $J$. Also, $\lambda_H(J)=1$ if $J\cong H$. If otherwise $J$ is isomorphic to some graph in $\calS_H$ but $J\ncong H$, we have
\begin{equation}\label{equ:A2}
\lambda_H(J)=-\sum_{\substack{H'\in \calS_H \\ \text{s.t. }H'\ncong H}} \mu_H(H')\lambda_{H'}(J).
\end{equation}

In order to verify~\eqref{equ:A1}, we have to determine $\lambda_H(J)$ for all $J\in S_H$. As $\lambda_H(J)$ is defined inductively by~\eqref{equ:A2}, we first determine $\lambda_{H'}(J)$ for all $H'\in \calS_H$ with $H'\ncong H$.

We start with the graph $H_{10}$ and determine $\lambda_{H_{10}}$. Clearly, $H_{10}$ has only one connected subgraph and we can choose $\calS_{H_{10}}=\{H_{10}\}$. Recall that $\lambda_{H_{10}}(J)=0$ for all graphs $J$ that are not isomorphic to any graph in $\calS_{H_{10}}$, i.e.\ not isomorphic to $H_{10}$ in this case. By definition we have 
\[\mu_{H_{10}}(H_{10})=1 \quad \text{as well as}\quad \lambda_{H_{10}}(H_{10})=1, \text{ see Table~\ref{tab:H10decomp}}.\]
This conforms with our intuition as for the single vertex graph $H_{10}$, it clearly holds that 
\begin{equation}\label{equ:A3}
\comp{G}{H_{10}} =\hom{G}{H_{10}}.
\end{equation}
Thus, we have now verified~\eqref{equ:A1} for $H=H_{10}$.

Using this information, we consider the graph $H_9$ next and determine $\mu_{H_9}$ and $\lambda_{H_9}$ for $\calS_{H_9}=\{H_9,H_{10}\}$, see Table~\ref{tab:H9decomp}. $H_9$ contains two connected subgraphs that are isomorphic to $H_{10}$, therefore $\mu_{H_9}(H_{10}) = 2$. Then, from~\eqref{equ:A2} we obtain \[\lambda_{H_9}(H_{10}) = -\sum_{H'\in \{H_{10}\}} \mu_{H_9}(H')\lambda_{H'}(H_{10})=-2.\] Plugging this into~\eqref{equ:A1} for $H=H_9$, we get
\begin{align}
\comp{G}{H_9} &= \sum_{J\in \calS_{H_9}} \lambda_{H_9}(J)\hom{G}{J}\nonumber\\
&=\hom{G}{H_9} - 2 \hom{G}{H_{10}}. \label{equ:A4}
\end{align}

Now let us verify this expression. Recall that $G$ is connected. The central idea behind our approach is that every homomorphism from $G$ to $H_9$ is a compaction onto some connected subgraph $H'$ of $H_9$. Furthermore, $\mu_{H_9}(H')$ tells us how many such subgraphs there are that are isomorphic to $H'$. Thus,
\begin{align*}
\hom{G}{H_9} &=\mu_{H_9}(H_{9})\cdot\comp{G}{H_9} + \mu_{H_9}(H_{10})\cdot\comp{G}{H_{10}}\\
&= \comp{G}{H_9} + 2 \comp{G}{H_{10}}.
\end{align*}
Rearranging and using the fact that we already know $\comp{G}{H_{10}} =\hom{G}{H_{10}}$ from~\eqref{equ:A3}:
\begin{align*}
\comp{G}{H_9} &= \hom{G}{H_9} - 2 \comp{G}{H_{10}}\\
&= \hom{G}{H_9} - 2 \hom{G}{H_{10}}.
\end{align*}
Thus, we have now proved~\eqref{equ:A4} which in turn proves~\eqref{equ:A1} for $H=H_9$.

Using~\eqref{equ:A3} and~\eqref{equ:A4} we can now go on to find (see Table~\ref{tab:H8decomp}) that 
\[\comp{G}{H_8}= \hom{G}{H_8} - 2 \hom{G}{H_{9}}+ \hom{G}{H_{10}}\]
and so on.

This gives the intuition behind the formal definitions of $\mu_H$ and $\lambda_H$. For completeness, we give the values for all graphs $H_1$ through $H_{10}$ in Tables~\ref{tab:H10decomp} through~\ref{tab:K23decomp}. From Table~\ref{tab:K23decomp} we can conclude that for $H=K_{2,3}$ the statement of Theorem~\ref{thm:CompDecomp} gives
\begin{align*}
\comp{G}{K_{2,3}} &= \hom{G}{K_{2,3}} - 6 \hom{G}{H_2} + 6 \hom{G}{H_3}\\
&\hphantom{=} + 3 \hom{G}{H_4} + 6 \hom{G}{H_5} - 2 \hom{G}{H_6}\\
&\hphantom{=} - 12\hom{G}{H_7} + 3 \hom{G}{H_8}.
\end{align*}

\begin{table}[h]
\centering
\def\scaleFactor{.6}
\def\arraystretch{1.25}
\begin{tabularx}{\textwidth}{@{}c|*5{>{\centering\arraybackslash}X}@{}}
$H'$& $H_{10}$ & & & &\\
\hline

&
\begin{tikzpicture}[scale=\scaleFactor, baseline=\scaleFactor*(2cm-.2\baselineskip)]

			\coordinate (L1)  at (0, 3);
			\coordinate (L2)  at (0, 1);
			\coordinate (R1)  at (3, 4);
			\coordinate (R2)  at (3, 2);
			\coordinate (R3)  at (3, 0);
			
			\fill (L1) circle[radius=3pt];
			
			\addvmargin{3mm}
\end{tikzpicture}
&
&
&
&
\\
\hline
$\mu_{H_{10}}(H')$& $1$ & & & &\\
\hline
$\lambda_{H_{10}}(H')$& $1$ & & & &
\end{tabularx}
\caption{Decomposition of $H_{10}$}
\label{tab:H10decomp}
\end{table}

\begin{table}
\centering
\def\scaleFactor{.6}
\def\arraystretch{1.25}
\begin{tabularx}{\textwidth}{@{}c|*5{>{\centering\arraybackslash}X}@{}}
$H'$& $H_{9}$ & $H_{10}$ & & &\\
\hline
 
&
\begin{tikzpicture}[scale=\scaleFactor, baseline=\scaleFactor*(2cm-.2\baselineskip)]

			\coordinate (L1)  at (0, 3);
			\coordinate (L2)  at (0, 1);
			\coordinate (R1)  at (3, 4);
			\coordinate (R2)  at (3, 2);
			\coordinate (R3)  at (3, 0);
			
			\fill (L1) circle[radius=3pt];
			\fill (R1) circle[radius=3pt];
	
			\draw (R1) -- (L1);
			\addvmargin{3mm}
\end{tikzpicture}
&
\begin{tikzpicture}[scale=\scaleFactor, baseline=\scaleFactor*(2cm-.2\baselineskip)]

			\coordinate (L1)  at (0, 3);
			\coordinate (L2)  at (0, 1);
			\coordinate (R1)  at (3, 4);
			\coordinate (R2)  at (3, 2);
			\coordinate (R3)  at (3, 0);
			
			\fill (L1) circle[radius=3pt];
			
			\addvmargin{3mm}
\end{tikzpicture}
&
&
&
\\
\hline
$\mu_{H_{9}}(H')$& $1$ & $2$ & & &\\
\hline
$\lambda_{H_{9}}(H')$& $1$ & $-2$ & & &
\end{tabularx}
\caption{Decomposition of $H_9$}
\label{tab:H9decomp}
\end{table}

\begin{table}
\centering
\def\scaleFactor{.6}
\def\arraystretch{1.25}
\begin{tabularx}{\textwidth}{@{}c|*5{>{\centering\arraybackslash}X}@{}}
$H'$& $H_{8}$ & $H_{9}$ & $H_{10}$ & &\\
\hline
 
&
\begin{tikzpicture}[scale=\scaleFactor, baseline=\scaleFactor*(2cm-.2\baselineskip)]

			\coordinate (L1)  at (0, 3);
			\coordinate (L2)  at (0, 1);
			\coordinate (R1)  at (3, 4);
			\coordinate (R2)  at (3, 2);
			\coordinate (R3)  at (3, 0);
			
			\fill (L1) circle[radius=3pt];
			\fill (R1) circle[radius=3pt];
			\fill (R2) circle[radius=3pt];
	
			\draw (R1) -- (L1) -- (R2);
			\addvmargin{3mm}
\end{tikzpicture}
&
\begin{tikzpicture}[scale=\scaleFactor, baseline=\scaleFactor*(2cm-.2\baselineskip)]

			\coordinate (L1)  at (0, 3);
			\coordinate (L2)  at (0, 1);
			\coordinate (R1)  at (3, 4);
			\coordinate (R2)  at (3, 2);
			\coordinate (R3)  at (3, 0);
			
			\fill (L1) circle[radius=3pt];
			\fill (R1) circle[radius=3pt];
	
			\draw (R1) -- (L1);
			\addvmargin{3mm}
\end{tikzpicture}
&
\begin{tikzpicture}[scale=\scaleFactor, baseline=\scaleFactor*(2cm-.2\baselineskip)]

			\coordinate (L1)  at (0, 3);
			\coordinate (L2)  at (0, 1);
			\coordinate (R1)  at (3, 4);
			\coordinate (R2)  at (3, 2);
			\coordinate (R3)  at (3, 0);
			
			\fill (L1) circle[radius=3pt];
			
			\addvmargin{3mm}
\end{tikzpicture}
&
&
\\
\hline
$\mu_{H_{8}}(H')$& $1$ & $2$ & $3$ & &\\
\hline
$\lambda_{H_{8}}(H')$& $1$ & $-2$ & $1$ & &
\end{tabularx}
\caption{Decomposition of $H_8$}
\label{tab:H8decomp}
\end{table}

\begin{table}
\centering
\def\scaleFactor{.6}
\def\arraystretch{1.25}
\begin{tabularx}{\textwidth}{@{}c|*5{>{\centering\arraybackslash}X}@{}}
$H'$& $H_{7}$ & $H_{8}$ & $H_{9}$ &$H_{10}$&\\
\hline
&
\begin{tikzpicture}[scale=\scaleFactor, baseline=\scaleFactor*(2cm-.2\baselineskip)]

			\coordinate (L1)  at (0, 3);
			\coordinate (L2)  at (0, 1);
			\coordinate (R1)  at (3, 4);
			\coordinate (R2)  at (3, 2);
			\coordinate (R3)  at (3, 0);
			
			\fill (L1) circle[radius=3pt];
			\fill (L2) circle[radius=3pt];
			\fill (R2) circle[radius=3pt];
			\fill (R3) circle[radius=3pt];
	
			\draw (R2) -- (L1) -- (R3) -- (L2);
			\addvmargin{3mm}
\end{tikzpicture}
&
\begin{tikzpicture}[scale=\scaleFactor, baseline=\scaleFactor*(2cm-.2\baselineskip)]

			\coordinate (L1)  at (0, 3);
			\coordinate (L2)  at (0, 1);
			\coordinate (R1)  at (3, 4);
			\coordinate (R2)  at (3, 2);
			\coordinate (R3)  at (3, 0);
			
			\fill (L1) circle[radius=3pt];
			\fill (R1) circle[radius=3pt];
			\fill (R2) circle[radius=3pt];
	
			\draw (R1) -- (L1) -- (R2);
			\addvmargin{3mm}
\end{tikzpicture}
&
\begin{tikzpicture}[scale=\scaleFactor, baseline=\scaleFactor*(2cm-.2\baselineskip)]

			\coordinate (L1)  at (0, 3);
			\coordinate (L2)  at (0, 1);
			\coordinate (R1)  at (3, 4);
			\coordinate (R2)  at (3, 2);
			\coordinate (R3)  at (3, 0);
			
			\fill (L1) circle[radius=3pt];
			\fill (R1) circle[radius=3pt];
	
			\draw (R1) -- (L1);
			\addvmargin{3mm}
\end{tikzpicture}
&
\begin{tikzpicture}[scale=\scaleFactor, baseline=\scaleFactor*(2cm-.2\baselineskip)]

			\coordinate (L1)  at (0, 3);
			\coordinate (L2)  at (0, 1);
			\coordinate (R1)  at (3, 4);
			\coordinate (R2)  at (3, 2);
			\coordinate (R3)  at (3, 0);
			
			\fill (L1) circle[radius=3pt];
			
			\addvmargin{3mm}
\end{tikzpicture}
&
\\
\hline
$\mu_{H_{7}}(H')$& $1$ & $2$ & $3$ & $4$ &\\
\hline
$\lambda_{H_{7}}(H')$& $1$ & $-2$ & $1$ &$0$&
\end{tabularx}
\caption{Decomposition of $H_7$}
\label{tab:H7decomp}
\end{table}

\begin{table}
\centering
\def\scaleFactor{.6}
\def\arraystretch{1.25}
\begin{tabularx}{\textwidth}{@{}c|*5{>{\centering\arraybackslash}X}@{}}
$H'$& $H_{6}$ & $H_{8}$ & $H_{9}$ &$H_{10}$&\\
\hline 
&
\begin{tikzpicture}[scale=\scaleFactor, baseline=\scaleFactor*(2cm-.2\baselineskip)]

			\coordinate (L1)  at (0, 3);
			\coordinate (L2)  at (0, 1);
			\coordinate (R1)  at (3, 4);
			\coordinate (R2)  at (3, 2);
			\coordinate (R3)  at (3, 0);
			
			\fill (L1) circle[radius=3pt];
			\fill (R1) circle[radius=3pt];
			\fill (R2) circle[radius=3pt];
			\fill (R3) circle[radius=3pt];
	
			\draw (R1) -- (L1) -- (R2);
			\draw (L1) -- (R3);
			\addvmargin{3mm}
\end{tikzpicture}
&
\begin{tikzpicture}[scale=\scaleFactor, baseline=\scaleFactor*(2cm-.2\baselineskip)]

			\coordinate (L1)  at (0, 3);
			\coordinate (L2)  at (0, 1);
			\coordinate (R1)  at (3, 4);
			\coordinate (R2)  at (3, 2);
			\coordinate (R3)  at (3, 0);
			
			\fill (L1) circle[radius=3pt];
			\fill (R1) circle[radius=3pt];
			\fill (R2) circle[radius=3pt];
	
			\draw (R1) -- (L1) -- (R2);
			\addvmargin{3mm}
\end{tikzpicture}
&
\begin{tikzpicture}[scale=\scaleFactor, baseline=\scaleFactor*(2cm-.2\baselineskip)]

			\coordinate (L1)  at (0, 3);
			\coordinate (L2)  at (0, 1);
			\coordinate (R1)  at (3, 4);
			\coordinate (R2)  at (3, 2);
			\coordinate (R3)  at (3, 0);
			
			\fill (L1) circle[radius=3pt];
			\fill (R1) circle[radius=3pt];
	
			\draw (R1) -- (L1);
			\addvmargin{3mm}
\end{tikzpicture}
&
\begin{tikzpicture}[scale=\scaleFactor, baseline=\scaleFactor*(2cm-.2\baselineskip)]

			\coordinate (L1)  at (0, 3);
			\coordinate (L2)  at (0, 1);
			\coordinate (R1)  at (3, 4);
			\coordinate (R2)  at (3, 2);
			\coordinate (R3)  at (3, 0);
			
			\fill (L1) circle[radius=3pt];
			
			\addvmargin{3mm}
\end{tikzpicture}
&
\\
\hline
$\mu_{H_{6}}(H')$& $1$ & $3$ & $3$ & $4$ &\\
\hline
$\lambda_{H_{6}}(H')$& $1$ & $-3$ & $3$ &$-1$&
\end{tabularx}
\caption{Decomposition of $H_6$}
\label{tab:H6decomp}
\end{table}

\begin{table}
\centering
\def\scaleFactor{.6}
\def\arraystretch{1.25}
\begin{tabularx}{\textwidth}{@{}c|*5{>{\centering\arraybackslash}X}@{}}
$H'$& $H_{5}$ & $H_{6}$ & $H_{7}$ &$H_{8}$& $H_{9}$\\
\hline
&
\begin{tikzpicture}[scale=\scaleFactor, baseline=\scaleFactor*(2cm-.2\baselineskip)]

			\coordinate (L1)  at (0, 3);
			\coordinate (L2)  at (0, 1);
			\coordinate (R1)  at (3, 4);
			\coordinate (R2)  at (3, 2);
			\coordinate (R3)  at (3, 0);
			
			\fill (L1) circle[radius=3pt];
			\fill (L2) circle[radius=3pt];
			\fill (R1) circle[radius=3pt];
			\fill (R2) circle[radius=3pt];
			\fill (R3) circle[radius=3pt];
	
			\draw (R1) -- (L1) -- (R2);
			\draw (L2) -- (R3);
			\draw (L1) -- (R3);
			\addvmargin{3mm}
\end{tikzpicture}
&
\begin{tikzpicture}[scale=\scaleFactor, baseline=\scaleFactor*(2cm-.2\baselineskip)]

			\coordinate (L1)  at (0, 3);
			\coordinate (L2)  at (0, 1);
			\coordinate (R1)  at (3, 4);
			\coordinate (R2)  at (3, 2);
			\coordinate (R3)  at (3, 0);
			
			\fill (L1) circle[radius=3pt];
			\fill (R1) circle[radius=3pt];
			\fill (R2) circle[radius=3pt];
			\fill (R3) circle[radius=3pt];
	
			\draw (R1) -- (L1) -- (R2);
			\draw (L1) -- (R3);
			\addvmargin{3mm}
\end{tikzpicture}
&
\begin{tikzpicture}[scale=\scaleFactor, baseline=\scaleFactor*(2cm-.2\baselineskip)]

			\coordinate (L1)  at (0, 3);
			\coordinate (L2)  at (0, 1);
			\coordinate (R1)  at (3, 4);
			\coordinate (R2)  at (3, 2);
			\coordinate (R3)  at (3, 0);
			
			\fill (L1) circle[radius=3pt];
			\fill (L2) circle[radius=3pt];
			\fill (R2) circle[radius=3pt];
			\fill (R3) circle[radius=3pt];
	
			\draw (R2) -- (L1) -- (R3) -- (L2);
			\addvmargin{3mm}
\end{tikzpicture}
&
\begin{tikzpicture}[scale=\scaleFactor, baseline=\scaleFactor*(2cm-.2\baselineskip)]

			\coordinate (L1)  at (0, 3);
			\coordinate (L2)  at (0, 1);
			\coordinate (R1)  at (3, 4);
			\coordinate (R2)  at (3, 2);
			\coordinate (R3)  at (3, 0);
			
			\fill (L1) circle[radius=3pt];
			\fill (R1) circle[radius=3pt];
			\fill (R2) circle[radius=3pt];
	
			\draw (R1) -- (L1) -- (R2);
			\addvmargin{3mm}
\end{tikzpicture}
&
\begin{tikzpicture}[scale=\scaleFactor, baseline=\scaleFactor*(2cm-.2\baselineskip)]

			\coordinate (L1)  at (0, 3);
			\coordinate (L2)  at (0, 1);
			\coordinate (R1)  at (3, 4);
			\coordinate (R2)  at (3, 2);
			\coordinate (R3)  at (3, 0);
			
			\fill (L1) circle[radius=3pt];
			\fill (R1) circle[radius=3pt];
	
			\draw (R1) -- (L1);
			\addvmargin{3mm}
\end{tikzpicture}
\\
\hline
$\mu_{H_{5}}(H')$& $1$ & $1$ & $2$ & $4$ & $4$\\
\hline
$\lambda_{H_{5}}(H')$& $1$ & $-1$ & $-2$ &$3$& $-1$\\
\multicolumn{6}{c}{}\\
$H'$& $H_{10}$ & & & &\\
\hline
&
\begin{tikzpicture}[scale=\scaleFactor, baseline=\scaleFactor*(2cm-.2\baselineskip)]

			\coordinate (L1)  at (0, 3);
			\coordinate (L2)  at (0, 1);
			\coordinate (R1)  at (3, 4);
			\coordinate (R2)  at (3, 2);
			\coordinate (R3)  at (3, 0);
			
			\fill (L1) circle[radius=3pt];
			
			\addvmargin{3mm}
\end{tikzpicture}
&

&
&
&
\\
\hline
$\mu_{H_{5}}(H')$& $5$ & & & &\\
\hline
$\lambda_{H_{5}}(H')$& $0$ &  & & &
\end{tabularx}
\caption{Decomposition of $H_5$}
\label{tab:H5decomp}
\end{table}

\begin{table}
\centering
\def\scaleFactor{.6}
\def\arraystretch{1.25}
\begin{tabularx}{\textwidth}{@{}c|*5{>{\centering\arraybackslash}X}@{}}
$H'$& $H_{4}$ & $H_{7}$ & $H_{8}$ &$H_{9}$& $H_{10}$\\
\hline 
&
\begin{tikzpicture}[scale=\scaleFactor, baseline=\scaleFactor*(2cm-.2\baselineskip)]

			\coordinate (L1)  at (0, 3);
			\coordinate (L2)  at (0, 1);
			\coordinate (R1)  at (3, 4);
			\coordinate (R2)  at (3, 2);
			\coordinate (R3)  at (3, 0);
			
			\fill (L1) circle[radius=3pt];
			\fill (L2) circle[radius=3pt];
			\fill (R2) circle[radius=3pt];
			\fill (R3) circle[radius=3pt];
	
			\draw (L1) -- (R2) -- (L2) -- (R3);
			\draw (L1) -- (R3);
			\addvmargin{3mm}
\end{tikzpicture}
&
\begin{tikzpicture}[scale=\scaleFactor, baseline=\scaleFactor*(2cm-.2\baselineskip)]

			\coordinate (L1)  at (0, 3);
			\coordinate (L2)  at (0, 1);
			\coordinate (R1)  at (3, 4);
			\coordinate (R2)  at (3, 2);
			\coordinate (R3)  at (3, 0);
			
			\fill (L1) circle[radius=3pt];
			\fill (L2) circle[radius=3pt];
			\fill (R2) circle[radius=3pt];
			\fill (R3) circle[radius=3pt];
	
			\draw (R2) -- (L1) -- (R3) -- (L2);
			\addvmargin{3mm}
\end{tikzpicture}
&
\begin{tikzpicture}[scale=\scaleFactor, baseline=\scaleFactor*(2cm-.2\baselineskip)]

			\coordinate (L1)  at (0, 3);
			\coordinate (L2)  at (0, 1);
			\coordinate (R1)  at (3, 4);
			\coordinate (R2)  at (3, 2);
			\coordinate (R3)  at (3, 0);
			
			\fill (L1) circle[radius=3pt];
			\fill (R1) circle[radius=3pt];
			\fill (R2) circle[radius=3pt];
	
			\draw (R1) -- (L1) -- (R2);
			\addvmargin{3mm}
\end{tikzpicture}
&
\begin{tikzpicture}[scale=\scaleFactor, baseline=\scaleFactor*(2cm-.2\baselineskip)]

			\coordinate (L1)  at (0, 3);
			\coordinate (L2)  at (0, 1);
			\coordinate (R1)  at (3, 4);
			\coordinate (R2)  at (3, 2);
			\coordinate (R3)  at (3, 0);
			
			\fill (L1) circle[radius=3pt];
			\fill (R1) circle[radius=3pt];
	
			\draw (R1) -- (L1);
			\addvmargin{3mm}
\end{tikzpicture}
&
\begin{tikzpicture}[scale=\scaleFactor, baseline=\scaleFactor*(2cm-.2\baselineskip)]

			\coordinate (L1)  at (0, 3);
			\coordinate (L2)  at (0, 1);
			\coordinate (R1)  at (3, 4);
			\coordinate (R2)  at (3, 2);
			\coordinate (R3)  at (3, 0);
			
			\fill (L1) circle[radius=3pt];
			
			\addvmargin{3mm}
\end{tikzpicture}
\\
\hline
$\mu_{H_{4}}(H')$& $1$ & $4$ & $4$ & $4$ & $4$\\
\hline
$\lambda_{H_{4}}(H')$& $1$ & $-4$ & $4$ & $0$ & $0$
\end{tabularx}
\caption{Decomposition of $H_4$}
\label{tab:H4decomp}
\end{table}

\begin{table}
\centering
\def\scaleFactor{.6}
\def\arraystretch{1.25}
\begin{tabularx}{\textwidth}{@{}c|*5{>{\centering\arraybackslash}X}@{}}
$H'$& $H_{3}$ & $H_{7}$ & $H_{8}$ &$H_{9}$& $H_{10}$\\
\hline 
&
\begin{tikzpicture}[scale=\scaleFactor, baseline=\scaleFactor*(2cm-.2\baselineskip)]

			\coordinate (L1)  at (0, 3);
			\coordinate (L2)  at (0, 1);
			\coordinate (R1)  at (3, 4);
			\coordinate (R2)  at (3, 2);
			\coordinate (R3)  at (3, 0);
			
			\fill (L1) circle[radius=3pt];
			\fill (L2) circle[radius=3pt];
			\fill (R1) circle[radius=3pt];
			\fill (R2) circle[radius=3pt];
			\fill (R3) circle[radius=3pt];
	
			\draw (R1) -- (L1) -- (R2) -- (L2) -- (R3);
			\addvmargin{3mm}
\end{tikzpicture}
&
\begin{tikzpicture}[scale=\scaleFactor, baseline=\scaleFactor*(2cm-.2\baselineskip)]

			\coordinate (L1)  at (0, 3);
			\coordinate (L2)  at (0, 1);
			\coordinate (R1)  at (3, 4);
			\coordinate (R2)  at (3, 2);
			\coordinate (R3)  at (3, 0);
			
			\fill (L1) circle[radius=3pt];
			\fill (L2) circle[radius=3pt];
			\fill (R2) circle[radius=3pt];
			\fill (R3) circle[radius=3pt];
	
			\draw (R2) -- (L1) -- (R3) -- (L2);
			\addvmargin{3mm}
\end{tikzpicture}
&
\begin{tikzpicture}[scale=\scaleFactor, baseline=\scaleFactor*(2cm-.2\baselineskip)]

			\coordinate (L1)  at (0, 3);
			\coordinate (L2)  at (0, 1);
			\coordinate (R1)  at (3, 4);
			\coordinate (R2)  at (3, 2);
			\coordinate (R3)  at (3, 0);
			
			\fill (L1) circle[radius=3pt];
			\fill (R1) circle[radius=3pt];
			\fill (R2) circle[radius=3pt];
	
			\draw (R1) -- (L1) -- (R2);
			\addvmargin{3mm}
\end{tikzpicture}
&
\begin{tikzpicture}[scale=\scaleFactor, baseline=\scaleFactor*(2cm-.2\baselineskip)]

			\coordinate (L1)  at (0, 3);
			\coordinate (L2)  at (0, 1);
			\coordinate (R1)  at (3, 4);
			\coordinate (R2)  at (3, 2);
			\coordinate (R3)  at (3, 0);
			
			\fill (L1) circle[radius=3pt];
			\fill (R1) circle[radius=3pt];
	
			\draw (R1) -- (L1);
			\addvmargin{3mm}
\end{tikzpicture}
&
\begin{tikzpicture}[scale=\scaleFactor, baseline=\scaleFactor*(2cm-.2\baselineskip)]

			\coordinate (L1)  at (0, 3);
			\coordinate (L2)  at (0, 1);
			\coordinate (R1)  at (3, 4);
			\coordinate (R2)  at (3, 2);
			\coordinate (R3)  at (3, 0);
			
			\fill (L1) circle[radius=3pt];
			
			\addvmargin{3mm}
\end{tikzpicture}
\\
\hline
$\mu_{H_{3}}(H')$& $1$ & $2$ & $3$ & $4$ & $5$\\
\hline
$\lambda_{H_{3}}(H')$& $1$ & $-2$ & $1$ & $0$ & $0$
\end{tabularx}
\caption{Decomposition of $H_3$}
\label{tab:H3decomp}
\end{table}

\begin{table}
\centering
\def\scaleFactor{.6}
\def\arraystretch{1.25}
\begin{tabularx}{\textwidth}{@{}c|*5{>{\centering\arraybackslash}X}@{}}
$H'$& $H_{2}$ & $H_{3}$ & $H_{4}$ &$H_{5}$& $H_{6}$\\
\hline  
&
\begin{tikzpicture}[scale=\scaleFactor, baseline=\scaleFactor*(2cm-.2\baselineskip)]

			\coordinate (L1)  at (0, 3);
			\coordinate (L2)  at (0, 1);
			\coordinate (R1)  at (3, 4);
			\coordinate (R2)  at (3, 2);
			\coordinate (R3)  at (3, 0);
			
			\fill (L1) circle[radius=3pt];
			\fill (L2) circle[radius=3pt];
			\fill (R1) circle[radius=3pt];
			\fill (R2) circle[radius=3pt];
			\fill (R3) circle[radius=3pt];
	
			\draw (R1) -- (L1) -- (R2) -- (L2) -- (R3);
			\draw (L1) -- (R3);
			\addvmargin{3mm}
\end{tikzpicture}
&
\begin{tikzpicture}[scale=\scaleFactor, baseline=\scaleFactor*(2cm-.2\baselineskip)]

			\coordinate (L1)  at (0, 3);
			\coordinate (L2)  at (0, 1);
			\coordinate (R1)  at (3, 4);
			\coordinate (R2)  at (3, 2);
			\coordinate (R3)  at (3, 0);
			
			\fill (L1) circle[radius=3pt];
			\fill (L2) circle[radius=3pt];
			\fill (R1) circle[radius=3pt];
			\fill (R2) circle[radius=3pt];
			\fill (R3) circle[radius=3pt];
	
			\draw (R1) -- (L1) -- (R2) -- (L2) -- (R3);
			\addvmargin{3mm}
\end{tikzpicture}
&
\begin{tikzpicture}[scale=\scaleFactor, baseline=\scaleFactor*(2cm-.2\baselineskip)]

			\coordinate (L1)  at (0, 3);
			\coordinate (L2)  at (0, 1);
			\coordinate (R1)  at (3, 4);
			\coordinate (R2)  at (3, 2);
			\coordinate (R3)  at (3, 0);
			
			\fill (L1) circle[radius=3pt];
			\fill (L2) circle[radius=3pt];
			\fill (R2) circle[radius=3pt];
			\fill (R3) circle[radius=3pt];
	
			\draw (L1) -- (R2) -- (L2) -- (R3);
			\draw (L1) -- (R3);
			\addvmargin{3mm}
\end{tikzpicture}
&
\begin{tikzpicture}[scale=\scaleFactor, baseline=\scaleFactor*(2cm-.2\baselineskip)]

			\coordinate (L1)  at (0, 3);
			\coordinate (L2)  at (0, 1);
			\coordinate (R1)  at (3, 4);
			\coordinate (R2)  at (3, 2);
			\coordinate (R3)  at (3, 0);
			
			\fill (L1) circle[radius=3pt];
			\fill (L2) circle[radius=3pt];
			\fill (R1) circle[radius=3pt];
			\fill (R2) circle[radius=3pt];
			\fill (R3) circle[radius=3pt];
	
			\draw (R1) -- (L1) -- (R2);
			\draw (L2) -- (R3);
			\draw (L1) -- (R3);
			\addvmargin{3mm}
\end{tikzpicture}
&
\begin{tikzpicture}[scale=\scaleFactor, baseline=\scaleFactor*(2cm-.2\baselineskip)]

			\coordinate (L1)  at (0, 3);
			\coordinate (L2)  at (0, 1);
			\coordinate (R1)  at (3, 4);
			\coordinate (R2)  at (3, 2);
			\coordinate (R3)  at (3, 0);
			
			\fill (L1) circle[radius=3pt];
			\fill (R1) circle[radius=3pt];
			\fill (R2) circle[radius=3pt];
			\fill (R3) circle[radius=3pt];
	
			\draw (R1) -- (L1) -- (R2);
			\draw (L1) -- (R3);
			\addvmargin{3mm}
\end{tikzpicture}
\\
\hline
$\mu_{H_{2}}(H')$& $1$ & $2$ & $1$ & $2$ & $1$\\
\hline
$\lambda_{H_{2}}(H')$& $1$ & $-2$ & $-1$ & $-2$ & $1$\\
\multicolumn{6}{c}{}\\
$H'$& $H_{7}$ & $H_{8}$ & $H_{9}$ &$H_{10}$& \\
\hline 
&
\begin{tikzpicture}[scale=\scaleFactor, baseline=\scaleFactor*(2cm-.2\baselineskip)]

			\coordinate (L1)  at (0, 3);
			\coordinate (L2)  at (0, 1);
			\coordinate (R1)  at (3, 4);
			\coordinate (R2)  at (3, 2);
			\coordinate (R3)  at (3, 0);
			
			\fill (L1) circle[radius=3pt];
			\fill (L2) circle[radius=3pt];
			\fill (R2) circle[radius=3pt];
			\fill (R3) circle[radius=3pt];
	
			\draw (R2) -- (L1) -- (R3) -- (L2);
			\addvmargin{3mm}
\end{tikzpicture}
&
\begin{tikzpicture}[scale=\scaleFactor, baseline=\scaleFactor*(2cm-.2\baselineskip)]

			\coordinate (L1)  at (0, 3);
			\coordinate (L2)  at (0, 1);
			\coordinate (R1)  at (3, 4);
			\coordinate (R2)  at (3, 2);
			\coordinate (R3)  at (3, 0);
			
			\fill (L1) circle[radius=3pt];
			\fill (R1) circle[radius=3pt];
			\fill (R2) circle[radius=3pt];
	
			\draw (R1) -- (L1) -- (R2);
			\addvmargin{3mm}
\end{tikzpicture}
&
\begin{tikzpicture}[scale=\scaleFactor, baseline=\scaleFactor*(2cm-.2\baselineskip)]

			\coordinate (L1)  at (0, 3);
			\coordinate (L2)  at (0, 1);
			\coordinate (R1)  at (3, 4);
			\coordinate (R2)  at (3, 2);
			\coordinate (R3)  at (3, 0);
			
			\fill (L1) circle[radius=3pt];
			\fill (R1) circle[radius=3pt];
	
			\draw (R1) -- (L1);
			\addvmargin{3mm}
\end{tikzpicture}
&
\begin{tikzpicture}[scale=\scaleFactor, baseline=\scaleFactor*(2cm-.2\baselineskip)]

			\coordinate (L1)  at (0, 3);
			\coordinate (L2)  at (0, 1);
			\coordinate (R1)  at (3, 4);
			\coordinate (R2)  at (3, 2);
			\coordinate (R3)  at (3, 0);
			
			\fill (L1) circle[radius=3pt];
			
			\addvmargin{3mm}
\end{tikzpicture}
&
\\
\hline
$\mu_{H_{2}}(H')$& $6$ & $6$ & $5$ & $5$ &\\
\hline
$\lambda_{H_{2}}(H')$& $6$ & $-3$ & $0$ & $0$ &
\end{tabularx}
\caption{Decomposition of $H_2$}
\label{tab:H2decomp}
\end{table}

\begin{table}[h]
\centering
\def\scaleFactor{.6}
\def\arraystretch{1.25}
\begin{tabularx}{\textwidth}{@{}c|*5{>{\centering\arraybackslash}X}@{}}
$H'$& $H_{1}$ & $H_{2}$ & $H_{3}$ &$H_{4}$& $H_{5}$\\
\hline  
& 
\begin{tikzpicture}[scale=\scaleFactor, baseline=\scaleFactor*(2cm-.2\baselineskip)]

			\coordinate (L1)  at (0, 3);
			\coordinate (L2)  at (0, 1);
			\coordinate (R1)  at (3, 4);
			\coordinate (R2)  at (3, 2);
			\coordinate (R3)  at (3, 0);
			
			\fill (L1) circle[radius=3pt];
			\fill (L2) circle[radius=3pt];
			\fill (R1) circle[radius=3pt];
			\fill (R2) circle[radius=3pt];
			\fill (R3) circle[radius=3pt];
	
			\draw (R1) -- (L1) -- (R2) -- (L2) -- (R3);
			\draw (L1) -- (R3);
			\draw (L2) -- (R1);
			\addvmargin{3mm}
\end{tikzpicture}
&
\begin{tikzpicture}[scale=\scaleFactor, baseline=\scaleFactor*(2cm-.2\baselineskip)]

			\coordinate (L1)  at (0, 3);
			\coordinate (L2)  at (0, 1);
			\coordinate (R1)  at (3, 4);
			\coordinate (R2)  at (3, 2);
			\coordinate (R3)  at (3, 0);
			
			\fill (L1) circle[radius=3pt];
			\fill (L2) circle[radius=3pt];
			\fill (R1) circle[radius=3pt];
			\fill (R2) circle[radius=3pt];
			\fill (R3) circle[radius=3pt];
	
			\draw (R1) -- (L1) -- (R2) -- (L2) -- (R3);
			\draw (L1) -- (R3);
			\addvmargin{3mm}
\end{tikzpicture}
&
\begin{tikzpicture}[scale=\scaleFactor, baseline=\scaleFactor*(2cm-.2\baselineskip)]

			\coordinate (L1)  at (0, 3);
			\coordinate (L2)  at (0, 1);
			\coordinate (R1)  at (3, 4);
			\coordinate (R2)  at (3, 2);
			\coordinate (R3)  at (3, 0);
			
			\fill (L1) circle[radius=3pt];
			\fill (L2) circle[radius=3pt];
			\fill (R1) circle[radius=3pt];
			\fill (R2) circle[radius=3pt];
			\fill (R3) circle[radius=3pt];
	
			\draw (R1) -- (L1) -- (R2) -- (L2) -- (R3);
			\addvmargin{3mm}
\end{tikzpicture}
&
\begin{tikzpicture}[scale=\scaleFactor, baseline=\scaleFactor*(2cm-.2\baselineskip)]

			\coordinate (L1)  at (0, 3);
			\coordinate (L2)  at (0, 1);
			\coordinate (R1)  at (3, 4);
			\coordinate (R2)  at (3, 2);
			\coordinate (R3)  at (3, 0);
			
			\fill (L1) circle[radius=3pt];
			\fill (L2) circle[radius=3pt];
			\fill (R2) circle[radius=3pt];
			\fill (R3) circle[radius=3pt];
	
			\draw (L1) -- (R2) -- (L2) -- (R3);
			\draw (L1) -- (R3);
			\addvmargin{3mm}
\end{tikzpicture}
&
\begin{tikzpicture}[scale=\scaleFactor, baseline=\scaleFactor*(2cm-.2\baselineskip)]

			\coordinate (L1)  at (0, 3);
			\coordinate (L2)  at (0, 1);
			\coordinate (R1)  at (3, 4);
			\coordinate (R2)  at (3, 2);
			\coordinate (R3)  at (3, 0);
			
			\fill (L1) circle[radius=3pt];
			\fill (L2) circle[radius=3pt];
			\fill (R1) circle[radius=3pt];
			\fill (R2) circle[radius=3pt];
			\fill (R3) circle[radius=3pt];
	
			\draw (R1) -- (L1) -- (R2);
			\draw (L2) -- (R3);
			\draw (L1) -- (R3);
			\addvmargin{3mm}
\end{tikzpicture}
\\
\hline
$\mu_{H_{1}}(H')$& $1$ & $6$ & $6$ & $3$ & $6$\\
\hline
$\lambda_{H_{1}}(H')$& $1$ & $-6$ & $6$ & $3$ & $6$\\
\multicolumn{6}{c}{}\\
$H'$& $H_{6}$ & $H_{7}$ & $H_{8}$ &$H_{9}$& $H_{10}$\\
\hline  
&
\begin{tikzpicture}[scale=\scaleFactor, baseline=\scaleFactor*(2cm-.2\baselineskip)]

			\coordinate (L1)  at (0, 3);
			\coordinate (L2)  at (0, 1);
			\coordinate (R1)  at (3, 4);
			\coordinate (R2)  at (3, 2);
			\coordinate (R3)  at (3, 0);
			
			\fill (L1) circle[radius=3pt];
			\fill (R1) circle[radius=3pt];
			\fill (R2) circle[radius=3pt];
			\fill (R3) circle[radius=3pt];
	
			\draw (R1) -- (L1) -- (R2);
			\draw (L1) -- (R3);
			\addvmargin{3mm}
\end{tikzpicture}
&
\begin{tikzpicture}[scale=\scaleFactor, baseline=\scaleFactor*(2cm-.2\baselineskip)]

			\coordinate (L1)  at (0, 3);
			\coordinate (L2)  at (0, 1);
			\coordinate (R1)  at (3, 4);
			\coordinate (R2)  at (3, 2);
			\coordinate (R3)  at (3, 0);
			
			\fill (L1) circle[radius=3pt];
			\fill (L2) circle[radius=3pt];
			\fill (R2) circle[radius=3pt];
			\fill (R3) circle[radius=3pt];
	
			\draw (R2) -- (L1) -- (R3) -- (L2);
			\addvmargin{3mm}
\end{tikzpicture}
&
\begin{tikzpicture}[scale=\scaleFactor, baseline=\scaleFactor*(2cm-.2\baselineskip)]

			\coordinate (L1)  at (0, 3);
			\coordinate (L2)  at (0, 1);
			\coordinate (R1)  at (3, 4);
			\coordinate (R2)  at (3, 2);
			\coordinate (R3)  at (3, 0);
			
			\fill (L1) circle[radius=3pt];
			\fill (R1) circle[radius=3pt];
			\fill (R2) circle[radius=3pt];
	
			\draw (R1) -- (L1) -- (R2);
			\addvmargin{3mm}
\end{tikzpicture}
&
\begin{tikzpicture}[scale=\scaleFactor, baseline=\scaleFactor*(2cm-.2\baselineskip)]

			\coordinate (L1)  at (0, 3);
			\coordinate (L2)  at (0, 1);
			\coordinate (R1)  at (3, 4);
			\coordinate (R2)  at (3, 2);
			\coordinate (R3)  at (3, 0);
			
			\fill (L1) circle[radius=3pt];
			\fill (R1) circle[radius=3pt];
	
			\draw (R1) -- (L1);
			\addvmargin{3mm}
\end{tikzpicture}
&
\begin{tikzpicture}[scale=\scaleFactor, baseline=\scaleFactor*(2cm-.2\baselineskip)]

			\coordinate (L1)  at (0, 3);
			\coordinate (L2)  at (0, 1);
			\coordinate (R1)  at (3, 4);
			\coordinate (R2)  at (3, 2);
			\coordinate (R3)  at (3, 0);
			
			\fill (L1) circle[radius=3pt];
			
			\addvmargin{3mm}
\end{tikzpicture}
\\
\hline
$\mu_{H_{1}}(H')$& $2$ & $12$ & $9$ & $6$ & $5$\\
\hline
$\lambda_{H_{1}}(H')$& $-2$ & $-12$ & $3$ & $0$ & $0$
\end{tabularx}
\caption{Decomposition of $H_1=K_{2,3}$}
\label{tab:K23decomp}
\end{table}

\end{document}